\documentclass[journal]{IEEEtran}
\usepackage{amsmath, amssymb, amsfonts, mathrsfs}
\usepackage{amsthm, mathtools, bm, dsfont}

\usepackage{tabularx, multirow, booktabs}

\usepackage[nocompress]{cite}
\usepackage{hyperref}

\usepackage{graphicx}
\makeatletter
    \let\MYcaption\@makecaption
\makeatother
    \usepackage[font=footnotesize]{subcaption}
\makeatletter
    \let\@makecaption\MYcaption
\makeatother

\usepackage{tikz}
\usetikzlibrary{shapes, shapes.geometric, arrows.meta, positioning, intersections, calc}
\tikzset{>=latex}

\DeclareMathOperator*{\argmin}{arg\,min}

\newcommand{\set}[1]{\mathcal{#1}}
\newtheorem{theorem}{Theorem}
\newtheorem{lemma}{Lemma}
\newtheorem{remark}{Remark}
\newtheorem{definition}{Definition}

\newtheorem{assumption}{Assumption}
\newtheorem{corollary}{Corollary}

\allowdisplaybreaks
\usepackage{flushend}

\begin{document}
\title{On the Cost of Consecutive Estimation Error: Significance-Aware Non-linear Aging}
\author{Jiping~Luo,~Nikolaos~Pappas,~\IEEEmembership{Senior~Member,~IEEE}
\thanks{This work has been supported in part by the Swedish Research Council (VR), ELLIIT, the Graduate School in Computer Science (CUGS), the European Union (ETHER, 101096526, ROBUST-6G, 101139068, and 6G-LEADER, 101192080), and the European Union's Horizon Europe research and innovation programme under the Marie Skłodowska-Curie Grant Agreement No 101131481 (SOVEREIGN). This paper was presented in part at the WiOpt 2025~\cite{jiping2025WiOpt}.}
\thanks{The authors are with the Department of Computer and Information Science, Link\"oping University, Link\"oping 58183, Sweden (e-mail: jiping.luo@liu.se; nikolaos.pappas@liu.se).}
}
\maketitle
\begin{abstract}
    This paper considers the semantics-aware remote state estimation of an asymmetric Markov chain with \emph{prioritized} states. Due to resource constraints, the sensor needs to trade off estimation quality against communication cost. The aim is to exploit the \emph{significance} of information through the history of system realizations to determine the optimal timing of transmission, thereby reducing the amount of uninformative data transmitted in the network. To this end, we introduce a new metric, the \emph{significance-aware Age of Consecutive Error} (AoCE), that captures three semantic attributes: the \emph{significance of estimation error}, the \emph{cost of consecutive error} (or \emph{lasting impact}, for short), and the \emph{urgency of lasting impact}. Different costs and non-linear age functions are assigned to different estimation errors to account for their relative importance to system performance. We identify the optimal transmission problem as a countably infinite state Markov decision process (MDP) with unbounded costs. We first give sufficient conditions on the age functions, source pattern, and channel reliability so that an optimal policy exists to have bounded average costs. We show that the optimal policy exhibits a \emph{switching structure}. That is, the sensor triggers a transmission only when the system has been trapped in an error for a certain number of consecutive time slots. We also provide sufficient conditions under which the switching policy degenerates into a simple \emph{threshold policy}, i.e., featuring identical thresholds for all estimation errors. Furthermore, we exploit the structural results and develop a \emph{structured policy iteration} (SPI) algorithm that considerably reduces computation overhead. Numerical results show that the optimal policy outperforms the classic rule-, distortion- and age-based policies. An important takeaway is that \emph{the more semantic attributes we utilize, the fewer transmissions are needed.}
\end{abstract}

\section{Introduction}
\IEEEPARstart{R}{emote} state estimation is a fundamental and significant problem in networked control systems (NCSs)\cite{Brockett-TAC-1997, Hespanha-ProcIEEE-2007, Schenato-ProcIEEE-2007, WNCSSurvey-2018}. Such systems often involve battery-powered devices sending local observations to remote ends over bandwidth-limited networks. Therefore, the transmitter can only transmit intermittently to trade off estimation quality against resource utilization~\cite{LingShi-TAC-2013, Nayyar-TAC-2013, Dey-TAC-2017, Aditya-TAC-2017, SunYin-TIT-2020, luo2025TCOM}. An important question arises: How should the transmitter determine which measurements are valuable?

In classical remote estimation, estimation quality is measured by distortion metrics such as Hamming distortion or mean square error, where a measurement is valuable if it contributes to a more \emph{accurate} estimate at the receiver\cite{LingShi-TAC-2013, Nayyar-TAC-2013, Dey-TAC-2017, Aditya-TAC-2017}. The underlying assumption is that \emph{all source states convey equally important information, and the cost of estimation error depends solely on the discrepancy between the source and the reconstructed signal.} However, this assumption deserves a careful re-examination in many applications. For example, in manufacturing systems, a plant may either reside in a normal state or shift to an alarm state upon an abnormal change in the operation point\cite{VVV-JSAIT-2021, Vikram-TIT-2013, VVV-TIT-2010}. In this context, the alarm state is of greater importance, and consequently, missed alarms incur significantly higher costs than false alarms. Similarly, connected autonomous vehicles demand more precise status information in critical situations (e.g., off-track or dense traffic). This motivates us to revisit the definition of ``estimation quality" and incorporate \emph{data significance} into system design. 

\begin{figure}[t]
    \centering
    \includegraphics[width=\linewidth]{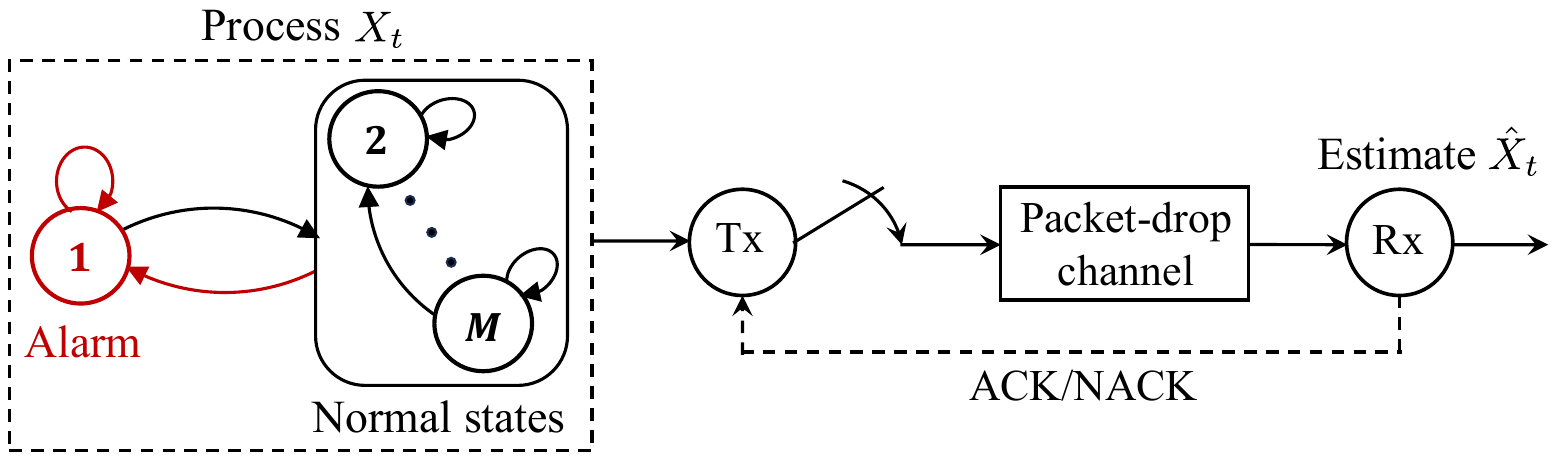}
    \caption{Remote estimation of a Markov source with prioritized states.}
    \label{fig:system-model}
\end{figure}

This paper considers the remote estimation of a finite-state Markov chain $\{X_t\}_{t\geq 1}$ (see Fig.~\ref{fig:system-model}). The sensor observes the chain and decides when to send measurements to the receiver. The receiver is tasked with constructing the estimate process $\{\hat{X}_t\}_{t\geq 1}$ based on the received measurements. Each state, labeled as $1, 2, \ldots, M$, can either represent a quantized level of a physical process or an abstract status of the system\footnote{In manufacturing systems, the normal and alarm states correspond to the pre- and post-change distributions of the underlying process, respectively.}. Given that some states convey more important information than others, we utilize three semantic\footnote{An introduction to semantics-aware communications can be found in~\cite{Marios}.} attributes to capture data significance: the \emph{significance of estimation error}, the \emph{cost of consecutive error} (or \emph{lasting impact}, for short), and the \emph{urgency of lasting impact}. The first attribute is represented by the content-aware distortion (see, e.g., \cite{Nikos-CAE-2021, Mehrdad-TCOM-2024, Mehrdad-JCN-2023, Zakeri-Asilomar-2023, luo2025TCOM}), where the cost of estimation error depends not only on the physical discrepancy but also on the contextual relevance and
potential control risks to system performance. The second attribute is motivated by the observation that the longer an error persists, the more severe its consequences can become~\cite{George-GCW-2019, Maatouk-TON-AoII-2020, Mehrdad-JCN-2023, luo2024MobiHoc}. Moreover, the urgency of lasting impact depends on the significance of estimation error. For example, autonomous vehicles might tolerate a moderate number of consecutive status errors under relatively safe conditions. However, inaccurate information about urgent states, even for a few seconds, can lead to wrong operations or even crashes. Existing metrics cannot adequately capture these semantic attributes. In this paper, we propose using a \emph{non-linear age function} $g_{X_t, \hat{X}_t}(\Delta_t)$ to model the cost of being in estimation error $(X_t, \hat{X}_t)$ for $\Delta_t$ consecutive time slots up until time $t$. Recently, information aging has received significant attention in remote estimation systems. However, most existing studies have been devoted to content-agnostic age metrics. The design of efficient policies for optimizing significance-aware non-linear aging remains largely unexplored. 

The main contributions of this work are as follows:
\begin{itemize}
    \item We introduce a new metric, the \emph{significance-aware Age of Consecutive Error} (AoCE), which accounts for both the significance of the current estimation error and the history-dependent cost of consecutive error. In essence, this metric assigns different costs and age functions to different estimation errors, offering the flexibility to handle the lasting impact of each error separately. For example, one might impose higher costs and exponential age functions on missed alarms while applying lower costs and logarithmic age functions to false alarms. 
    \item The remote estimation problem is formulated as a \emph{countably infinite} state Markov decision process (MDP) with \emph{unbounded} costs. We give sufficient conditions on the source pattern, age functions, and channel reliability that yield a deterministic optimal policy for this MDP. We prove that the optimal policy exhibits a \emph{switching structure}, as depicted in Fig.~\ref{fig:policy-structure}. That is, the sensor triggers a transmission only when the age of the error $(X_t, \hat{X}_t) = (i,j)$ exceeds a fixed threshold $\tau^*_{i,j}$. We also give conditions under which the optimal policy degenerates into a simple \emph{threshold policy}; that is, it has identical thresholds for all errors. 
    \item Our switching policy is significant in several aspects. It circumvents the ``curse of memory" and the ``curse of dimensionality" of the MDP. One only needs to compute a small number of threshold values offline and store them in the sensor memory instead of solving a high-dimensional dynamic programming recursion and saving the results for infinitely many states. Moreover, it answers the fundamental question of ``what and when to transmit". According to \cite{luo2025TCOM}, distortion-optimal policies only tell ``whether to transmit when a certain error occurs". That is, the optimal threshold for the error $(X_t, \hat{X}_t) = (i,j)$ is either $\tau^*_{i,j} = 1$ (i.e., always transmit) or $\tau^*_{i,j} \rightarrow \infty$ (i.e., never transmit). By incorporating error holding time as the third dimension in the decision-making process, our approach further determines the \emph{optimal timing} to initiate a transmission, allowing for transmissions to occur after several consecutive errors, i.e., $\tau_{i,j}^*\geq 1$. 
    \item For numerical tractability, we propose a state-space truncation method and show the asymptotic optimality of the truncated MDP. We exploit these findings to develop a \emph{structured policy iteration} (SPI) algorithm to compute the switching policy with reduced computation overhead. Our numerical results show that the switching policy can be much better than the classic rule-, distortion-, and age-based policies. This highlights that the significance-aware AoCE offers more informed decisions and extends the current understanding of distortion and information aging.
\end{itemize}

\begin{figure}
    \centering
    \includegraphics[width=0.95\linewidth]{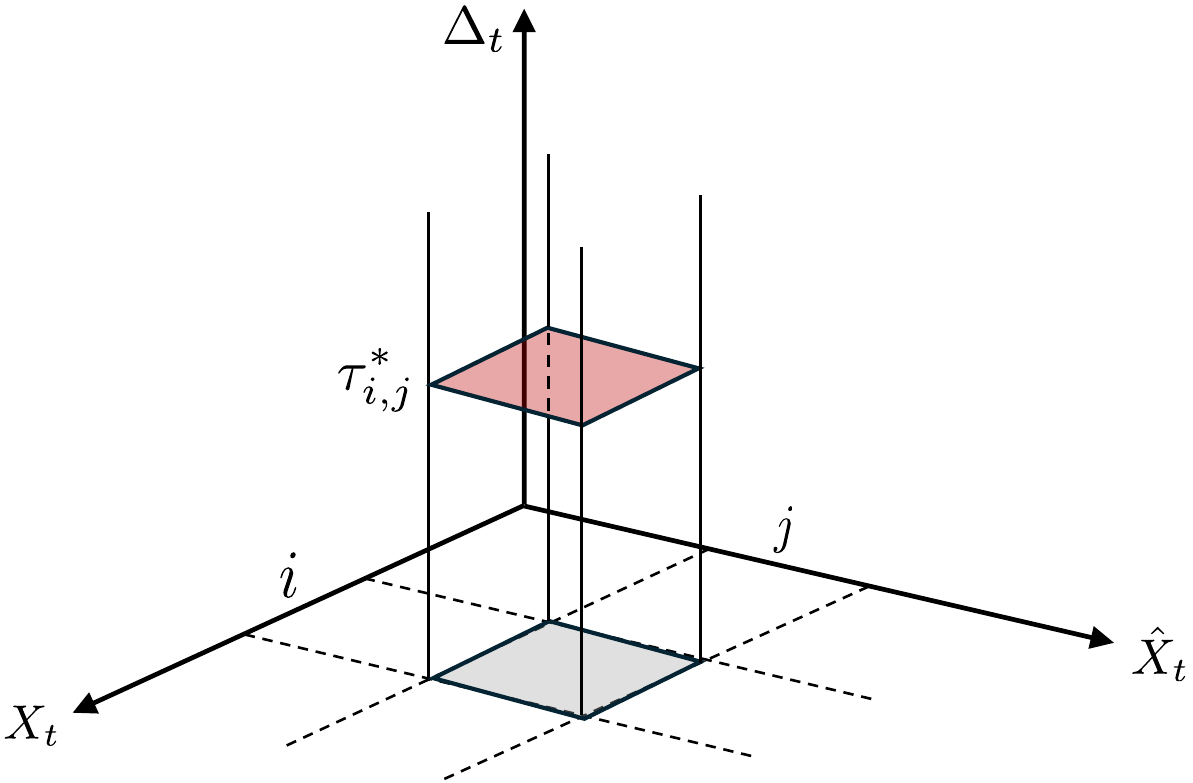}
    \caption{Switching structure of the optimal policy.}
    \label{fig:policy-structure}
\end{figure}

The rest of this paper is organized as follows. Related work is summarized in Section~\ref{sec:related-work}. Section~\ref{sec:system-model} introduces the system model and formulates the optimal transmission problem. Section~\ref{sec:main-results} presents our main results, including the structure of the optimal policy, the asymptotic optimality of the truncated MDP, and the SPI algorithm. The numerical results and the conclusion are provided in Section~\ref{sec:numerical-results} and Section~~\ref{sec:conclusion}.

\section{Related Work}\label{sec:related-work}
In the remote estimation literature, the primary objective has been to minimize distortion over constraints on available resources, such as channel bandwidth or energy budget. Since distortion depends on the physical discrepancy between the original and reconstructed signals, the optimal transmission and estimation policies must be \emph{signal-aware} and derived by accounting for the source's evolution pattern. Numerous studies have been devoted to \emph{linear Gaussian systems}. Remote estimation of a scalar Gaussian source with communication costs was studied in\cite{Lipsa-TAC-2011-scalar-source}, where the authors proved that a threshold transmission policy (i.e., the sensor transmits whenever the current error exceeds a threshold) and a Kalman-like estimator are jointly optimal. These results were further extended to systems with multidimensional Gaussian sources and energy-harvesting sensors\cite{Nayyar-TAC-2013}, hard constraints on transmission frequency\cite{Aditya-TAC-2017}, unreliable channels with adaptive noise\cite{gao2018optimal} and packet drops\cite{Aditya-TAC-2020}, to name a few. 

Another line of research focuses on the use of smart sensors (e.g., with Kalman filters) to pre-estimate source states and then send the estimated states, rather than raw measurements, to the receiver~\cite{LingShi-TAC-2013, Dey-TAC-2017, LingShi-2018-Automatica, WuShuang-TCNS-2020, Wanchun-TAC-2022}. Besides tractability, this approach establishes a connection between distortion and the Age of Information\footnote{Let $U_t$ be the generation time of the newest measurement received at the receiver by time $t$. The AoI is defined as the time elapsed since the latest measurement was generated\cite{Roy2012INFOCOM, kosta2017age}, i.e., $t - U_t$.} (AoI). It has been shown that the error covariance is a monotonic non-decreasing function of AoI\footnote{The monotonicity can be generalized to $p$-th order autoregressive $\textrm{AR}(p)$ linear processes, provided that the sensor sends a sequence of measurements no shorter than the process order $p$\cite{SunYin-INFOCOM-2024}.}. Consequently, AoI serves as a sufficient statistic for decision-making, and the optimal policy initiates a transmission whenever the age exceeds a threshold. These results suggest that, in such systems, measurements are more valuable when they are \emph{fresh}. However, this may not always be the case, as AoI ignores the source pattern and is, therefore, \emph{signal-agnostic}. For instance, in the remote estimation of Wiener and Ornstein-Uhlenbeck processes\cite{SunYin-TIT-2020, SunYin-TON-2021}, the estimation error achieved by the distortion-optimal policy can be much smaller than that of the age-optimal policy. 

Remote estimation of discrete-state \emph{Markov chains} has gained significant interest in recent years\cite{Aditya-TAC-2020, LingShi-TAC-HMM-2017, LingShi-TAC-2023-Neuromorphic-Camera}. A salient feature of Markov processes is that they evolve in a probabilistic manner; consequently, the estimation error does not necessarily evolve monotonically with AoI. Another reason traditional distortion and AoI metrics become inefficient in Markovian systems is that the states often convey richer information beyond simple physical amplitudes \cite{luo2025TCOM, luo2024MobiHoc}. Such motivated, the concept of semantics-aware estimation and a series of content-aware metrics that go beyond information accuracy and freshness have been proposed. Content-aware distortion was first introduced in\cite{Nikos-CAE-2021} to represent the significance of different estimation errors. Performance analysis and comparison of different policies were studied in\cite{Mehrdad-JCN-2023, Mehrdad-TCOM-2024, Zakeri-Asilomar-2023, luo2025TCOM}. Structural results of the optimal policy in resource-constrained systems were established in \cite{luo2025TCOM}. Various age metrics have been proposed to address the shortcomings of AoI. Content-aware AoI\cite{George-GCW-2019} and the Uncertainty of Information (UoI)\cite{gongpu2022UoI} reveal that the information quality depends on its content and evolves at different rates. The Version Age of Information (VAoI) tracks the number of content changes and is more relevant than AoI in Markovian systems\cite{yates2021versionage, sennur2022versionage, salimnejad2024version}. The Age of Incorrect Information (AoII) is a signal-aware metric that counts only the time elapsed since the system was last synced\cite{Maatouk-TON-AoII-2020, ChenYutao-TON-AoII-2024, Nail-INFOCOM-2024, Nail-ISIT-2024}. The optimality of threshold policies for symmetric Markov chains was established in \cite{Maatouk-TON-AoII-2020}. The authors of~\cite{Nail-ISIT-2024} extended these results to general continuous-time Markov chains, showing that the AoII-optimal policy has a switching structure that depends on the instantaneous estimation error. However, these age metrics treat all estimation errors equally, leading to inadequate transmissions in alarm states but excessive transmissions in normal states.

The closest study to this paper is our previous work in~\cite{luo2024MobiHoc}, where we introduced the Age of Missed Alarm (AoMA) and the Age of False Alarm (AoFA) to account for the lasting impact of a binary Markov chain. This paper generalizes~\cite{luo2024MobiHoc} in the following aspects: i) In this paper, we introduce a change-aware age metric, termed AoCE, that resets upon error variations, based on which we extend the results to general finite-state Markov chains. ii) While~\cite{luo2024MobiHoc} showed that an optimal policy always exists for linear age metrics, this does not hold for non-linear age functions. In this paper, we give sufficient conditions on the age functions, source pattern, and channel reliability so that an optimal policy exists to have bounded average costs. iii) The proofs in~\cite{luo2024MobiHoc} relied on analytical expressions of the switching policy, which do not apply to multi-state sources and non-linear age functions. We can generalize the theoretical results and greatly simplify the proofs by adopting new proof techniques. iv) In~\cite{luo2024MobiHoc}, the optimal policy was obtained by exhaustively searching for two optimal thresholds using analytical results. In contrast, this paper presents an SPI algorithm that exploits the structural properties to reduce computation overhead.

\section{System Model and Problem Formulation}\label{sec:system-model}
\subsection{Remote Estimation Model}
Consider the remote state estimation model illustrated in Fig.~\ref{fig:system-model}. The system comprises four main components: an information source, a local sensor (transmitter), a remotely placed estimator (receiver), and a wireless channel. 

The stochastic process considered is a finite-state, homogeneous, discrete-time Markov chain (DTMC) $\{X_t\}_{t\geq 1}$ defined on the finite state space
\begin{align}
    \set{X} = \{1, 2, \ldots, M\}.
\end{align}
Here, state $1$ is labeled as the ``alarm" state. We note that the results in this paper extend automatically to models with multiple alarm states. For later reference, we distinguish between the following types of estimation errors:
\begin{itemize}
    \item Missed alarms occur when the receiver falsely announces a normal state while the source is actually in the alarm state, i.e., $X_t = 1, \hat{X}_t\neq 1$. Timely detection of abnormalities is crucial for decision-making and system maintenance\cite{VVV-JSAIT-2021}.
   \item False alarms refer to erroneously raising an alarm at the receiver when the source is in a normal state, i.e., $X_t\neq 1, \hat{X}_t = 1$. Although less critical than missed alarms, false alarms can lead to unnecessary expenditure on checking the system thus wasting resources.
   \item Other normal errors are considered indistinguishable and are not of primary interest.
\end{itemize}

Let $Q$ denote the state transition probability matrix, where
\begin{align}
    Q = (Q_{i,j}, i,j\in\set{X}), ~
    Q_{i,j} = \Pr[X_{t+1}=j|X_t=i].
\end{align}
To avoid pathological cases, we assume $Q$ is irreducible. Let $\set{X}_\textrm{ap}$ and $\set{X}_\textrm{p}$ denote the sets of states with and without self-transitions, respectively, where
\begin{align}
    \set{X}_\textrm{ap} = \{i\in\set{X}:Q_{i,i}>0\}, ~\set{X}_\textrm{p} = \set{X}-\set{X}_\textrm{ap}.
\end{align}
If there is at least one self-transition, i.e., $\set{X}_\textrm{ap} \neq \emptyset$, then $Q$ is aperiodic. Unless otherwise stated, we assume $\set{X}_\textrm{ap} \neq \emptyset$.

The sensor sequentially observes the source state $X_t$ and decides at every decision epoch (i.e., the beginning of time slot $t$) whether or not to transmit a new measurement. Let $A_t\in \{0, 1\}$ denote the decision variable, where $A_t = 1$ means transmission while $A_t = 0$ means no transmission. We consider an error-prone channel with \emph{i.i.d.} packet drops. Let $H_t\in\{0, 1\}$ denote the packet dropout process, which is an \emph{i.i.d.} Bernoulli process satisfying
\begin{align}
    \Pr[H_t = 1] = p_s, ~
    \Pr[H_t = 0] = 1 - p_s = p_f.
\end{align}
Here, $H_t =1$ means the packet arrives at the destination by the end of time slot $t$ (i.e., the decision epoch $t+1$ in continuous time), whereas $H_t = 0$ indicates deep fading channel conditions, resulting in transmission failure.

Upon successful reception, the receiver updates its estimate using the latest received measurement\footnote{This is a common assumption in the literature\cite{Nikos-CAE-2021, Maatouk-TON-AoII-2020, luo2025TCOM}. The optimal estimate depends on the source statistics, transmission policy, and history of received measurements\cite{krishnamurthy2016POMDP}. However, this is out of the scope of this work.}, i.e., $\hat{X}_{t+1} = X_t$, and sends an acknowledgment (ACK) packet to the sensor. Otherwise, a negative ACK (NACK) is feedback, and the remote estimate remains unchanged, i.e., $\hat{X}_{t+1} = \hat{X}_t$. We assume that ACK/NACK packets are delivered instantaneously and error-free. Therefore, the sensor knows precisely the remote estimate $\hat{X}_t$ at every decision epoch. The information available at the sensor up to time $t$ is 
\begin{align}
    I_t = (X_{1:t}, \hat{X}_{1:t}, A_{1:t-1}).
\end{align}
At decision epoch $t$, a decision $A_t$ is taken according to a \emph{transmission rule} $\pi_t$, i.e.,
\begin{align}
    A_t = \pi_t(I_t) = \pi_t(X_{1:t}, \hat{X}_{1:t}, A_{1:t-1}).
\end{align}
A \emph{transmission policy} is a sequence of transmission rules, i.e., $\pi=(\pi_1, \pi_2, \ldots)$. We call a policy \emph{stationary} if it employs the same rule at every epoch $t$. A policy is \emph{deterministic} if, given the history $I_t$, it selects an action with certainty. A \emph{randomized} policy specifies a probability distribution on the action space. 

\subsection{Significance of Estimation Error} \label{sec:significance of estimation errors}
In classical remote state estimation, \emph{whether a measurement is discarded or transmitted does not depend on the significance of the measurement.} A widely used performance metric for Markov chains is the Hamming distortion~\cite{Tamer-HammingDistortion-2005, Nail-BinaryFreshness-2024}, i.e.,
\begin{align}
    d(X_t, \hat{X}_t) = \mathds{1}\{X_t\neq \hat{X}_t\},\label{eq:classical-distortion}
\end{align}
where $\mathds{1}\{\cdot\}$ is the indicator function.

Recall that in our problem, the alarm state (labeled as state $1$) is of greater interest. Intuitively, missed alarms typically incur higher costs than other estimation errors. Therefore, we employ a \emph{content-aware distortion} metric that assigns different costs to different estimation errors\cite{Nikos-CAE-2021, luo2025TCOM}, defined as
\begin{align}
    \bar{d}(X_t, \hat{X}_t) \triangleq \begin{cases}
        D_{i,j}, &\text{if}~(X_t, \hat{X}_t)=(i,j),i\neq j,\\
        0, &\text{otherwise},
    \end{cases}\label{eq:new-distortion}
\end{align}
where $D_{i,j}>0$ represents the significance of error $(i, j)$.

A notable shortcoming of distortion lies in its \emph{history-independence}. Although the source evolution is Markovian, \emph{the value of information carried by the measurement depends on the history of past observations and decisions}. For instance, not only does the instantaneous estimation error matter but also how long the system has been trapped in this error, i.e., the cost of consecutive error (or \emph{lasting impact}, for short)\cite{luo2024MobiHoc}. 

\begin{definition}
In this paper, we introduce a new age process, termed Age of Consecutive Error (AoCE), to capture this history-dependent attribute, defined as
\begin{align}
    \Delta_t \triangleq 
    \begin{cases}
        \Delta_{t-1}+ 1, &\text{if}~X_{t}\neq\hat{X}_{t}, (X_{t},\hat{X}_t)=(X_{t-1},\hat{X}_{t-1}),\\
        1, &\text{if}~X_t\neq\hat{X}_t,(X_{t},\hat{X}_t)\neq(X_{t-1},\hat{X}_{t-1}),\\
        0, &\text{if}~X_t=\hat{X}_t.
    \end{cases}\label{eq:aoce}
\end{align}
Fig.~\ref{fig:age-figures} compares AoCE with typical distortion and age metrics.
\end{definition} 

\begin{figure*}[t]
    \centering
    \begin{subfigure}{0.32\linewidth}
        \centering
        \scalebox{0.75}{\begin{tikzpicture}[scale=1.0]
\draw[->] (0,0) -- (6.8,0) node[anchor=north] {$t$};
\draw[->] (0,0) -- (0,3.5);
\draw (0, 3.7) node {$d_t$};
\draw (-0.2, 0.5) node {$1$};

\fill (0,0)  circle[radius=1.5pt];
\fill (1,0)  circle[radius=1.5pt];
\fill (2,0)  circle[radius=1.5pt];
\fill (3,0)  circle[radius=1.5pt];
\fill (4,0)  circle[radius=1.5pt];
\fill (5,0)  circle[radius=1.5pt];
\fill (6,0)  circle[radius=1.5pt];

\draw	(0,0) node[anchor=north] {$t_0$}
        (1,0) node[anchor=north] {$t_1$}
		(2,0) node[anchor=north] {$t_2$}
		(3,0) node[anchor=north] {$t_3$}
		(4,0) node[anchor=north] {$t_4$}
		(5,0) node[anchor=north] {$t_5$}
        (6,0) node[anchor=north] {$t_{6}$};

\draw[dotted] (1,0) -- (1,3);
\draw[dotted] (2,0) -- (2,3);
\draw[dotted] (3,0) -- (3,3);
\draw[dotted] (4,0) -- (4,3);
\draw[dotted] (5,0) -- (5,3);
\draw[dotted] (6,0) -- (6,3);
              (6.55,-0.15) -- (6.55,0.15);

\draw	(0.5,3) node{{\scriptsize $X=1$}}
        (0.5,2.7) node{{\scriptsize $\hat{X}=2$}}
		(1.5,2.99) node{{\scriptsize $1$}}
        (1.5,2.65) node{{\scriptsize $2$}}
        (2.5,2.99) node{{\scriptsize $1$}}
        (2.5,2.65) node{{\scriptsize $2$}}
        (3.5,2.99) node{{\scriptsize $3$}}
        (3.5,2.65) node{{\scriptsize $1$}}
        (4.5,2.99) node{{\scriptsize $3$}}
        (4.5,2.65) node{{\scriptsize $1$}}
        (5.5,2.99) node{{\scriptsize $1$}}
        (5.5,2.65) node{{\scriptsize $1$}}
        (6.5,2.99) node{{\scriptsize $1$}}
        (6.5,2.65) node{{\scriptsize $1$}};

\draw[line width=0.5mm]  (0,0.5) to[bend right=0] (5,0.5); 
\draw[line width=0.5mm]  (5,0.5) to[bend left=0] (5,0); 

\draw[line width=0.5mm]  (5,0) to[bend left=0] (6.5,0); 

\draw[->,>=stealth]    (0,-0.5) -- (0,-0.8);
\draw[->,>=stealth]    (3,-0.5) -- (3,-0.8);
         
\end{tikzpicture}}
	\label{fig:distortion}
        \caption{Distortion}
    \end{subfigure}
    \hfill
    \begin{subfigure}{0.32\linewidth}
        \centering
        \scalebox{0.75}{\begin{tikzpicture}[scale=1.0]
\draw[->] (0,0) -- (6.8,0) node[anchor=north] {$t$};
\draw[->] (0,0) -- (0,3.5);
\draw (0, 3.7) node {$\bar{d}_t$};
\draw (-0.2, 0.5) node {$1$};

\fill (0,0)  circle[radius=1.5pt];
\fill (1,0)  circle[radius=1.5pt];
\fill (2,0)  circle[radius=1.5pt];
\fill (3,0)  circle[radius=1.5pt];
\fill (4,0)  circle[radius=1.5pt];
\fill (5,0)  circle[radius=1.5pt];
\fill (6,0)  circle[radius=1.5pt];

\draw	(0,0) node[anchor=north] {$t_0$}
        (1,0) node[anchor=north] {$t_1$}
		(2,0) node[anchor=north] {$t_2$}
		(3,0) node[anchor=north] {$t_3$}
		(4,0) node[anchor=north] {$t_4$}
		(5,0) node[anchor=north] {$t_5$}
        (6,0) node[anchor=north] {$t_{6}$};

\draw[dotted] (1,0) -- (1,3);
\draw[dotted] (2,0) -- (2,3);
\draw[dotted] (3,0) -- (3,3);
\draw[dotted] (4,0) -- (4,3);
\draw[dotted] (5,0) -- (5,3);
\draw[dotted] (6,0) -- (6,3);
\draw[dotted] (0,0.5) -- (6.5,0.5);

              (6.55,-0.15) -- (6.55,0.15);

\draw	(0.5,3) node{{\scriptsize $X=1$}}
        (0.5,2.7) node{{\scriptsize $\hat{X}=2$}}
		(1.5,2.99) node{{\scriptsize $1$}}
        (1.5,2.65) node{{\scriptsize $2$}}
        (2.5,2.99) node{{\scriptsize $1$}}
        (2.5,2.65) node{{\scriptsize $2$}}
        (3.5,2.99) node{{\scriptsize $3$}}
        (3.5,2.65) node{{\scriptsize $1$}}
        (4.5,2.99) node{{\scriptsize $3$}}
        (4.5,2.65) node{{\scriptsize $1$}}
        (5.5,2.99) node{{\scriptsize $1$}}
        (5.5,2.65) node{{\scriptsize $1$}}
        (6.5,2.99) node{{\scriptsize $1$}}
        (6.5,2.65) node{{\scriptsize $1$}};

 
\draw[line width=0.5mm]  (0,1) to[bend right=0] (3,1); 
\draw[line width=0.5mm]  (3,1) to[bend right=0] (3,0.2); 

\draw[line width=0.5mm]  (3,0.2) to[bend left=0] (5,0.2); 
\draw[line width=0.5mm]  (5,0.2) to[bend left=0] (5,0); 

\draw[line width=0.5mm]  (5,0) to[bend left=0] (6.5,0); 

\draw[->,>=stealth]    (0,-0.5) -- (0,-0.8);
\draw[->,>=stealth]    (3,-0.5) -- (3,-0.8);
         
\end{tikzpicture}}
	\label{fig:state-aware-distortion}
        \caption{Content-aware distortion}
    \end{subfigure}
    \hfill
    \begin{subfigure}{0.32\linewidth}
        \centering
        \scalebox{0.75}{\begin{tikzpicture}[scale=1.0]
\draw[->] (0,0) -- (6.8,0) node[anchor=north] {$t$};
\draw[->] (0,0) -- (0,3.5);
\draw (0, 3.7) node {$\Delta_t$};
\draw (-0.2, 0.5) node {$1$};

\fill (0,0)  circle[radius=1.5pt];
\fill (1,0)  circle[radius=1.5pt];
\fill (2,0)  circle[radius=1.5pt];
\fill (3,0)  circle[radius=1.5pt];
\fill (4,0)  circle[radius=1.5pt];
\fill (5,0)  circle[radius=1.5pt];
\fill (6,0)  circle[radius=1.5pt];

\draw	(0,0) node[anchor=north] {$t_0$}
        (1,0) node[anchor=north] {$t_1$}
		(2,0) node[anchor=north] {$t_2$}
		(3,0) node[anchor=north] {$t_3$}
		(4,0) node[anchor=north] {$t_4$}
		(5,0) node[anchor=north] {$t_5$}
        (6,0) node[anchor=north] {$t_{6}$};

\draw[dotted] (1,0) -- (1,3);
\draw[dotted] (2,0) -- (2,3);
\draw[dotted] (3,0) -- (3,3);
\draw[dotted] (4,0) -- (4,3);
\draw[dotted] (5,0) -- (5,3);
\draw[dotted] (6,0) -- (6,3);
\draw[dotted] (0,0.5) -- (6.5,0.5);
              (6.55,-0.15) -- (6.55,0.15);

\draw	(0.5,3) node{{\scriptsize $X=1$}}
        (0.5,2.7) node{{\scriptsize $\hat{X}=2$}}
		(1.5,2.99) node{{\scriptsize $1$}}
        (1.5,2.65) node{{\scriptsize $2$}}
        (2.5,2.99) node{{\scriptsize $1$}}
        (2.5,2.65) node{{\scriptsize $2$}}
        (3.5,2.99) node{{\scriptsize $3$}}
        (3.5,2.65) node{{\scriptsize $1$}}
        (4.5,2.99) node{{\scriptsize $3$}}
        (4.5,2.65) node{{\scriptsize $1$}}
        (5.5,2.99) node{{\scriptsize $1$}}
        (5.5,2.65) node{{\scriptsize $1$}}
        (6.5,2.99) node{{\scriptsize $1$}}
        (6.5,2.65) node{{\scriptsize $1$}};

\draw[line width=0.5mm]  (0,0.5) to[bend right=0] (3,2); 
\draw[line width=0.5mm]  (3,2) to[bend right=0] (3,0.5); 

\draw[line width=0.5mm]  (3,0.5) to[bend left=0] (6.5,2.25); 

\draw[->,>=stealth]    (0,-0.5) -- (0,-0.8);
\draw[->,>=stealth]    (3,-0.5) -- (3,-0.8);
         
\end{tikzpicture}}
	\label{fig:AoI}
        \caption{AoI}
    \end{subfigure}\\
    \begin{subfigure}{0.32\linewidth}
        \centering
        \scalebox{0.75}{\begin{tikzpicture}[scale=1.0]
\draw[->] (0,0) -- (6.8,0) node[anchor=north] {$t$};
\draw[->] (0,0) -- (0,3.5);
\draw (0, 3.7) node {$\Delta_t$};
\draw (-0.2, 0.5) node {$1$};

\fill (0,0)  circle[radius=1.5pt];
\fill (1,0)  circle[radius=1.5pt];
\fill (2,0)  circle[radius=1.5pt];
\fill (3,0)  circle[radius=1.5pt];
\fill (4,0)  circle[radius=1.5pt];
\fill (5,0)  circle[radius=1.5pt];
\fill (6,0)  circle[radius=1.5pt];

\draw	(0,0) node[anchor=north] {$t_0$}
        (1,0) node[anchor=north] {$t_1$}
		(2,0) node[anchor=north] {$t_2$}
		(3,0) node[anchor=north] {$t_3$}
		(4,0) node[anchor=north] {$t_4$}
		(5,0) node[anchor=north] {$t_5$}
        (6,0) node[anchor=north] {$t_{6}$};

\draw[dotted] (1,0) -- (1,3);
\draw[dotted] (2,0) -- (2,3);
\draw[dotted] (3,0) -- (3,3);
\draw[dotted] (4,0) -- (4,3);
\draw[dotted] (5,0) -- (5,3);
\draw[dotted] (6,0) -- (6,3);
\draw[dotted] (0,0.5) -- (6.5,0.5);

\draw	(0.5,3) node{{\scriptsize $X=1$}}
        (0.5,2.7) node{{\scriptsize $\hat{X}=2$}}
		(1.5,2.99) node{{\scriptsize $1$}}
        (1.5,2.65) node{{\scriptsize $2$}}
        (2.5,2.99) node{{\scriptsize $1$}}
        (2.5,2.65) node{{\scriptsize $2$}}
        (3.5,2.99) node{{\scriptsize $3$}}
        (3.5,2.65) node{{\scriptsize $1$}}
        (4.5,2.99) node{{\scriptsize $3$}}
        (4.5,2.65) node{{\scriptsize $1$}}
        (5.5,2.99) node{{\scriptsize $1$}}
        (5.5,2.65) node{{\scriptsize $1$}}
        (6.5,2.99) node{{\scriptsize $1$}}
        (6.5,2.65) node{{\scriptsize $1$}};

\draw[line width=0.5mm]  (0,0.5) to[bend right=0] (5,3); 
\draw[line width=0.5mm]  (5,3) to[bend right=0] (5,0); 
\draw[line width=0.5mm]  (5,0) to[bend left=0] (6.5,0); 

\draw[->,>=stealth]    (0,-0.5) -- (0,-0.8);
\draw[->,>=stealth]    (3,-0.5) -- (3,-0.8);
         
\end{tikzpicture}}
	\label{fig:AoII}
        \caption{AoII}
    \end{subfigure}
    \hfill
    \begin{subfigure}{0.32\linewidth}
        \centering
        \scalebox{0.75}{\begin{tikzpicture}[scale=1.0]
\draw[->] (0,0) -- (6.8,0) node[anchor=north] {$t$};
\draw[->] (0,0) -- (0,3.5);
\draw (0, 3.7) node {$\Delta_t$};
\draw (-0.2, 0.5) node {$1$};

\fill (0,0)  circle[radius=1.5pt];
\fill (1,0)  circle[radius=1.5pt];
\fill (2,0)  circle[radius=1.5pt];
\fill (3,0)  circle[radius=1.5pt];
\fill (4,0)  circle[radius=1.5pt];
\fill (5,0)  circle[radius=1.5pt];
\fill (6,0)  circle[radius=1.5pt];

\draw	(0,0) node[anchor=north] {$t_0$}
        (1,0) node[anchor=north] {$t_1$}
		(2,0) node[anchor=north] {$t_2$}
		(3,0) node[anchor=north] {$t_3$}
		(4,0) node[anchor=north] {$t_4$}
		(5,0) node[anchor=north] {$t_5$}
        (6,0) node[anchor=north] {$t_{6}$};

\draw[dotted] (1,0) -- (1,3);
\draw[dotted] (2,0) -- (2,3);
\draw[dotted] (3,0) -- (3,3);
\draw[dotted] (4,0) -- (4,3);
\draw[dotted] (5,0) -- (5,3);
\draw[dotted] (6,0) -- (6,3);
\draw[dotted] (0,0.5) -- (6.5,0.5);

\draw	(0.5,3) node{{\scriptsize $X=1$}}
        (0.5,2.7) node{{\scriptsize $\hat{X}=2$}}
		(1.5,2.99) node{{\scriptsize $1$}}
        (1.5,2.65) node{{\scriptsize $2$}}
        (2.5,2.99) node{{\scriptsize $1$}}
        (2.5,2.65) node{{\scriptsize $2$}}
        (3.5,2.99) node{{\scriptsize $3$}}
        (3.5,2.65) node{{\scriptsize $1$}}
        (4.5,2.99) node{{\scriptsize $3$}}
        (4.5,2.65) node{{\scriptsize $1$}}
        (5.5,2.99) node{{\scriptsize $1$}}
        (5.5,2.65) node{{\scriptsize $1$}}
        (6.5,2.99) node{{\scriptsize $1$}}
        (6.5,2.65) node{{\scriptsize $1$}};

\draw[line width=0.5mm]  (0,0.5) to[bend right=0] (3,2); 
\draw[line width=0.5mm]  (3,2) to[bend right=0] (3,0.5); 

\draw[line width=0.5mm]  (3,0.5) to[bend left=0] (5,1.5); 
\draw[line width=0.5mm]  (5,1.5) to[bend left=0] (5,0); 

\draw[line width=0.5mm]  (5,0) to[bend left=0] (6.5,0); 

\draw[->,>=stealth]    (0,-0.5) -- (0,-0.8);
\draw[->,>=stealth]    (3,-0.5) -- (3,-0.8);
         
\end{tikzpicture}}
	\label{fig:AoCE}
        \caption{AoCE}
    \end{subfigure}
    \hfill
    \begin{subfigure}{0.32\linewidth}
        \centering
        \scalebox{0.75}{\begin{tikzpicture}[scale=1.0]
\draw[->] (0,0) -- (6.8,0) node[anchor=north] {$t$};
\draw[->] (0,0) -- (0,3.5);
\draw (0.1, 3.7) node {$c(S_t)$};
\draw (-0.2, 0.5) node {$1$};

\fill (0,0)  circle[radius=1.5pt];
\fill (1,0)  circle[radius=1.5pt];
\fill (2,0)  circle[radius=1.5pt];
\fill (3,0)  circle[radius=1.5pt];
\fill (4,0)  circle[radius=1.5pt];
\fill (5,0)  circle[radius=1.5pt];
\fill (6,0)  circle[radius=1.5pt];

\draw	(0,0) node[anchor=north] {$t_0$}
        (1,0) node[anchor=north] {$t_1$}
		(2,0) node[anchor=north] {$t_2$}
		(3,0) node[anchor=north] {$t_3$}
		(4,0) node[anchor=north] {$t_4$}
		(5,0) node[anchor=north] {$t_5$}
        (6,0) node[anchor=north] {$t_6$};

\draw[dotted] (1,0) -- (1,3);
\draw[dotted] (2,0) -- (2,3);
\draw[dotted] (3,0) -- (3,3);
\draw[dotted] (4,0) -- (4,3);
\draw[dotted] (5,0) -- (5,3);
\draw[dotted] (6,0) -- (6,3);
\draw[dotted] (0,0.5) -- (6.5,0.5);
              (6.55,-0.15) -- (6.55,0.15);

\draw	(0.5,3) node{{\scriptsize $X=1$}}
        (0.5,2.7) node{{\scriptsize $\hat{X}=2$}}
		(1.5,2.99) node{{\scriptsize $1$}}
        (1.5,2.65) node{{\scriptsize $2$}}
        (2.5,2.99) node{{\scriptsize $1$}}
        (2.5,2.65) node{{\scriptsize $2$}}
        (3.5,2.99) node{{\scriptsize $3$}}
        (3.5,2.65) node{{\scriptsize $1$}}
        (4.5,2.99) node{{\scriptsize $3$}}
        (4.5,2.65) node{{\scriptsize $1$}}
        (5.5,2.99) node{{\scriptsize $1$}}
        (5.5,2.65) node{{\scriptsize $1$}}
        (6.5,2.99) node{{\scriptsize $1$}}
        (6.5,2.65) node{{\scriptsize $1$}};

\draw[line width=0.5mm]  (0,1) to[bend right=30] (3,2.9); 
\draw[line width=0.5mm]  (3,2.9) to[bend right=0] (3,0.7); 

\draw[line width=0.5mm]  (3,0.7) to[bend left=15] (5,1.4); 
\draw[line width=0.5mm]  (5,1.4) to[bend left=0] (5,0); 

\draw[line width=0.5mm]  (5,0) to[bend left=0] (6.5,0); 

\draw[->,>=stealth]    (0,-0.5) -- (0,-0.8);
\draw[->,>=stealth]    (3,-0.5) -- (3,-0.8); 
\end{tikzpicture}}
	\label{fig:nonlinearAoCE}
        \caption{Significance-aware AoCE}
    \end{subfigure}\\   
    \vspace{-0.1in}
    \caption{Illustration of typical distortion and age metrics, where $(1, 2)$ is a missed alarm, and $(3, 1)$ is a false alarm. New measurements are received at $t_0$ and $t_3$; however, the system remains erroneous due to changes in the source state. The system is automatically synced at $t_5$. Distortion metrics (a)-(b) evaluate only the current estimation error, whereas age-based metrics (c)-(f) account for the history of past observations. The AoI (c) ignores the source evolution and continues to grow even when the system is synced. The AoII (d) increases by $1$ whenever an estimation error occurs, while the AoCE (e) resets upon error variations. The significance-aware AoCE (f) assigns exponential age penalties to missed alarms and logarithmic penalties to false alarms, accounting for the urgency of the lasting impact of different estimation errors.}
    \label{fig:age-figures}
\end{figure*}

\begin{remark}
    We note that AoCE~\eqref{eq:aoce} is change-aware, as it resets upon error changes, whereas AoII is change-agnostic and increments by $1$ regardless of the error type. This feature allows us to reconstruct AoII from AoCE without additional information, but the reverse is not possible. Moreover, it offers the flexibility to handle the lasting impact of different estimation errors separately. However, age alone may not suffice, as it ignores the urgency of the current estimation error. This gives incentives to significance-aware age metrics.
\end{remark}

Let $S_t$ denote the system state at decision epoch $t$, where
\begin{align}
    S_t = (X_t, \hat{X}_t, \Delta_t).\label{eq:system-state}
\end{align}
The \emph{significance-aware AoCE} for estimation error $(X_t,\hat{X}_t)$ at decision epoch $t$ is defined as
\begin{align}
    c(S_t) = c(I_t) \triangleq \bar{d}(X_t, \hat{X}_t)\cdot g_{X_t, \hat{X}_t}(\Delta_t),\label{eq:cost-function}
\end{align}
where $(g_{i,j}(\cdot), i,j\in\set{X})$ are \emph{non-negative}, \emph{non-decreasing}, and possibly \emph{unbounded} age functions. $g_{i,j}(\delta)$ represents the cost of being in error $(i,j)$ for $\delta$ consecutive time slots. These age functions are quite general and may be discontinuous and non-convex. Given that no cost is incurred in synced states, we impose $g_{i, i}(\cdot) = 0$ for all $i \in\set{X}$. 

\begin{remark}\label{remark:age-process}
    The significance of information is represented by the content-aware distortion $\bar{d}$, the history-dependent lasting impact $\Delta_t$, and the non-linear age functions $g_{i,j}$. The system state \eqref{eq:system-state} can be interpreted as a collection of $M(M-1)$ dependent age processes, each corresponding to an estimation error (see Fig.~\ref{fig:state-evolution}). Notably, the AoCE \eqref{eq:aoce} is not a sufficient statistic\footnote{A process $\Gamma_t \subseteq I_t$ is called a sufficient statistic if there is no loss of optimality in using transmission rules of the form: $A_t = \pi_t(\Gamma_t)$\cite{mahajan2016decentralized}. In other words, it summarizes all relevant information about the history.} for decision-making unless Assumption \ref{assumption:source} holds. We formally state this finding in Lemma~\ref{lemma:sufficient-statistic}. This result can be generalized to any age process whose evolution depends on the source pattern. 
\end{remark}

\begin{assumption}\label{assumption:source}
    The source is non-prioritized, i.e., $D_{i,j}=D$ and $g_{i,j} = g$ for all $i\neq j$, and is symmetric with equal state change probabilities, i.e.,
    \begin{align}
        Q_{i,j} = \begin{cases}
            p,        &\text{if}~i\neq j,\\
            \bar{p},  &\text{otherwise},
        \end{cases}\label{eq:symmetric-source}
    \end{align}
    where $\bar{p} + (M-1)p = 1, 0 < p < 1$. In other words, all estimation errors contribute equally to the system in terms of both costs and occurrences.
\end{assumption}

\begin{lemma}\label{lemma:sufficient-statistic}
    The AoCE~\eqref{eq:aoce} is a sufficient statistic for the MDP only if Assumption~\ref{assumption:source} holds.
\end{lemma}
\begin{proof}
    See Appendix~\ref{proof:sufficient-statistic}.
\end{proof}

\vspace{-0.1in}
\subsection{System Evolution}\label{sec:system-evolution}
The system state $\{S_t\}_{t\geq 1}$ is a three-dimensional controlled Markov chain. It is possible to achieve the desired performance by controlling the chain's transition probabilities. Let 
\begin{align}
    P = (P_{s,s^\prime}(a), s,s^\prime\in\set{S}, a\in\set{A})\label{eq:trans-prob}
\end{align}
denote the transition probability matrix, where
\begin{align}
    P_{s,s^\prime}(a) = \Pr[S_{t+1}=s^\prime|S_t = s, A_t = a]
\end{align}
is the probability of transitioning to state $s^\prime$ at time $t+1$, given that the system is in state $s$ and action $a$ is taken at time $t$. 

\begin{figure*}[t]
    \begin{subfigure}{0.493\linewidth}
        \centering
        \includegraphics[width=\linewidth]{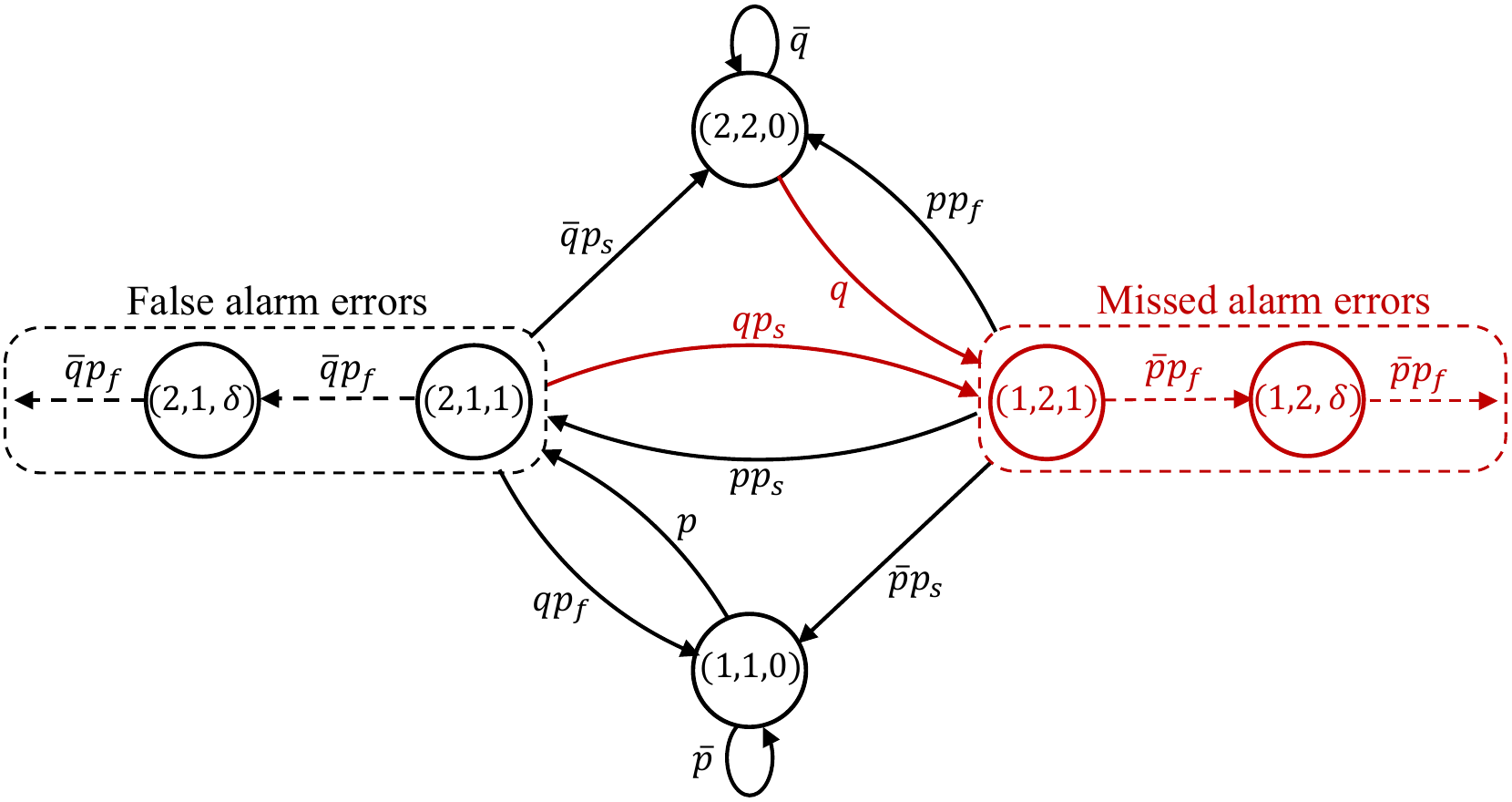}
        \caption{Always-transmit policy}
        \label{fig:state-DTMC-1}
    \end{subfigure}
    \hfill
    \begin{subfigure}{0.493\linewidth}
        \centering
        \includegraphics[width=\linewidth]{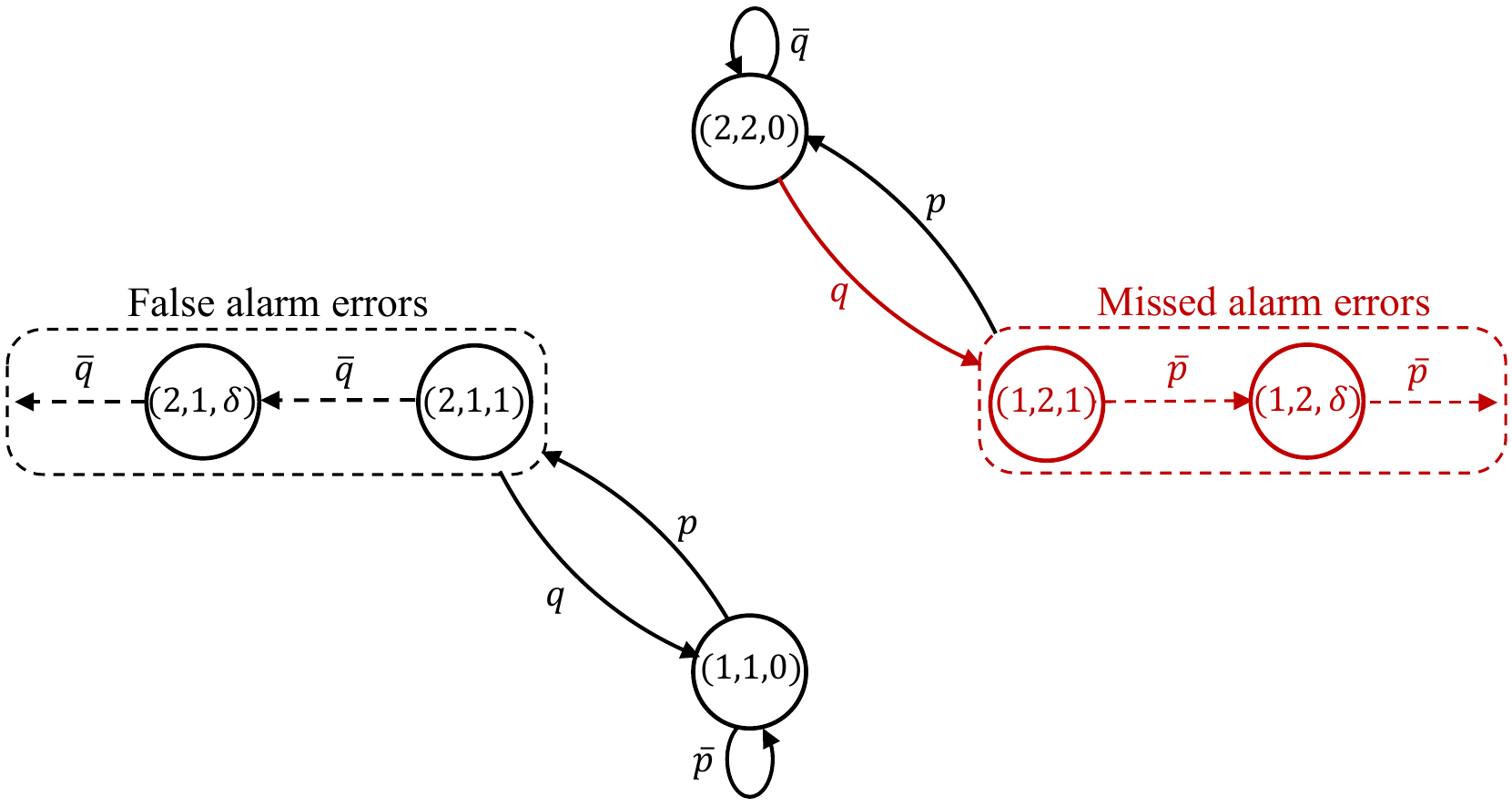}
        \caption{Non-transmit policy}
        \label{fig:state-DTMC-0}
    \end{subfigure}  
    \vspace{-0.1in}
    \caption{DTMC representing the system state evolution of a binary Markov chain, where $p=Q_{1,2}, q=Q_{2,1},\bar{p}=1-p,\bar{q}=1-q$.}
    \label{fig:state-evolution}
\end{figure*}

Fig.~\ref{fig:state-evolution} illustrates the system evolution of a binary chain. There are two dependent age processes interconnected through the synced states. The always-transmit policy (i.e., $A_t=1, t\geq 1$) induces an irreducible Markov chain, whereas the non-transmit policy (i.e., $A_t=0, t\geq 1$) divides the state space into two isolated groups. Below, we discuss the transition probabilities of a multi-state model.

For any estimation error $(i, j), i\neq j$ with an age $\delta\geq 1$, if the sensor decides not to transmit, the system will: (1) remain in this error, i.e., $s^\prime = (i, j, \delta+1)$, with probability (w.p.) $Q_{i, i}$, (2) become synced, i.e., $s^\prime = (j, j, 0)$, w.p. $Q_{i,j}$, or (3) change to another error $(k, j, 1), k\neq i,j$ w.p. $Q_{i, k}$. Thus, we have
\begin{align}
    \Pr[s^\prime|(i, j, \delta), 0] = 
    \begin{cases}
        Q_{i,i}, &\text{if}~ s^\prime = (i, j, \delta+1),\\
        Q_{i, j}, &\text{if}~s^\prime = (j, j, 0),\\
        Q_{i, k}, &\text{if}~s^\prime = (k, j, 1), k\neq i,j,\\
        0,  &\text{otherwise.}\label{eq:prob-error-0}
    \end{cases}
\end{align}

If the sensor initiates a transmission and the packet is successfully received at the receiver w.p. $p_s$, the system may either enter a synced state $(i, i, 0)$ w.p. $Q_{i, i}p_s$ or change to a new error $(k_1, i, 1), k_1\neq i$ w.p. $Q_{i, k_1}p_s$. On the other hand, if the transmission is unsuccessful w.p. $p_f$, the system will: (1) stay in the current error w.p. $Q_{i, i}p_f$, (2) become synced w.p. $Q_{i, j}p_f$, or (3) change to another error $(k_2, j), k_2\neq i,j$ w.p. $Q_{i, k_2}p_f$. Thus, we obtain
\begin{align}
    \Pr[s^\prime|(i, j, \delta), 1] = 
    \begin{cases}
        Q_{i, i}p_s, &\text{if}~s^\prime = (i, i, 0),\\
        Q_{i,k_1}p_s, &\text{if}~ s^\prime = (k_1, i, 1), k_1\neq i,\\
        Q_{i, i}p_f, &\text{if}~s^\prime = (i, j, \delta+1),\\
        Q_{i, j}p_f, &\text{if}~s^\prime = (j, j, 0),\\
        Q_{i,k_2}p_f, &\text{if}~ s^\prime = (k_2, j, 1), k_2\neq i,j,\\
        0,  &\text{otherwise.} \label{eq:prob-error-1}
    \end{cases}
\end{align}

For any synced state $(i, i, 0), \forall i\in \set{X}$, the remote estimate remains unchanged, regardless of whether a new measurement is received or not. Therefore, the system will either remain synced w.p. $Q_{i, i}$, or enter an estimation error $(k, i), k\neq i$ w.p. $Q_{i,k}$. Formally, for each $a\in\{0, 1\}$,
\begin{align}
    \Pr[s^\prime|(i, i, 0), a] = 
    \begin{cases}
        Q_{i,i}, &\text{if}~ s^\prime = (i, i, 0),\\
        Q_{i, k}, &\text{if}~s^\prime = (k, i, 1), k\neq i,\\
        0,  &\text{otherwise.}
    \end{cases}\label{eq:prob-synced-states}
\end{align}
It should be noted that, in all the above cases, the AoCE is no greater than $1$ if the source is in a state $i\in\set{X}_\textrm{p}$.

\begin{lemma}\label{lemma:recurrence}
     The DTMC $\{S_t\}_{t\geq 1}$ controlled by the always-transmit policy is irreducible and positive recurrent. Moreover, for every pair of states $s$ and $s^\prime$ in $\set{S}$, the expected first passage time from $s$ to $s^\prime$ is finite.
\end{lemma}
\begin{proof}
    See Appendix~\ref{proof:recurrence}.
\end{proof}

\subsection{Optimal Transmission Problem}
The goal is to achieve a desired balance between estimation performance and communication cost. Given the cost of each transmission $\lambda$, the (per-stage) cost of taking an action $A_t$ in state $S_t$ is given by
\begin{align}
    l(S_t, A_t) = c(S_{t})+ \lambda \mathds{1}\{A_t\neq 0\}.\label{eq:total-cost}
\end{align}
The expected average cost of a transmission policy $\pi$ over an infinite horizon is defined as
\begin{align}
    \set{L}(\pi) \triangleq \limsup_{T\rightarrow\infty}\frac{1}{T}\sum_{t=1}^{T}\mathbb{E}^{\pi} \bigl[l(S_t, A_t)\big|S_1 = s_1\bigr],\label{eq:average-cost}
\end{align}
where $\mathbb{E}^\pi$ represents the conditional expectation, given that policy $\pi$ is employed with initial state $s_1$. The sensor aims to determine the optimal policy $\pi^*$ to minimize \eqref{eq:average-cost}, i.e.,
\begin{align}
    \set{L}^* = \inf_{\pi\in\Pi} \set{L}(\pi),\label{problem:MDP}
\end{align}
where $\Pi$ is the set of all admissible policies. 

Problem \eqref{problem:MDP} is a Markov Decision Process (MDP) with an average cost optimality criterion. This MDP can be described by a tuple $(\set{S}, \set{A}, P, l)$, where $\set{S}$ is the state space of all possible values of $S_t$, $\set{A}=\{0, 1\}$ is the action space, $P$ is the transition probability matrix defined in \eqref{eq:trans-prob}, and $l$ is the cost function given by \eqref{eq:total-cost}. The state space $\set{S}$ is the union of the set of all synced states, $\set{S}_0$, and $M(M-1)$ sets of error holding times, $\set{S}_{i,j}$ for all $i\neq j$. These sets are defined as
\begin{align}
    \set{S}_0 &\triangleq \{(i,i,0): i\in\set{X}\},\\
    \set{S}_{i,j} &\triangleq \begin{cases}
        \{(i,j, \delta):\delta\geq 1\}, &\text{if}~i\neq j, i\in\set{X}_\textrm{ap},\\
        \{(i,j, 1)\}, &\text{if}~i\neq j, i\in\set{X}_\textrm{p}.
    \end{cases}
\end{align}

We note that $\set{S}$ is a \emph{countably infinite} set\footnote{A counterexample is when the source has no self-transitions, i.e., $\set{X}_\textrm{ap}=\emptyset$.} since the error holding time can grow indefinitely. Consequently, the problem encounters computing and memory challenges because classical dynamic programming methods cannot iterate over an infinite state space. Moreover, due to unbounded per-stage costs, an optimal policy may not exist\cite{sennott1998stochastic}. The following questions relating to the optimal policy are of interest:
\begin{enumerate}
    \item Under what conditions does an optimal policy exist?
    \item Are there any special structures of the optimal policy that facilitate policy implementation and computation?
    \item Is it possible to achieve asymptotic optimality by approximating the MDP with a finite state space?
    \item How to compute the optimal policy with reduced computation?
\end{enumerate}

\section{Main Results}\label{sec:main-results}
This section aims to answer the above questions. We first give sufficient conditions on the source pattern, age functions, and channel reliability that yield a deterministic optimal policy for the MDP. Next, we prove that the optimal policy exhibits a \emph{switching structure}, which facilitates policy storage and algorithm design. We also give sufficient conditions under which the optimal policy degenerates into a simple \emph{threshold policy}. For numerical tractability, we propose a state-space truncation method and show the asymptotic optimality of the truncated MDP. By exploiting these findings, we develop a structure-aware algorithm to compute this switching policy with reduced computation overhead.

\subsection{Existence of an Optimal Policy}
Recall that the AoCE is countably infinite, and the age functions are non-decreasing and unbounded. Consequently, the long-term average cost may never be bounded, no matter how we choose the transmission policy. Thus, we will be concerned with the conditions under which an optimal policy exists to achieve bounded average costs.  

Assumption~\ref{assumption:age-functions} gives such conditions, based on which we show the existence of an optimal policy in Theorem~\ref{theorem:existence}. This result reveals that the policy space $\Pi$ can be reduced to a small subset of Markovian (i.e., independent of $S_{1:t-1}$ and $A_{1:t-1}$), stationary, and deterministic policies without losing optimality. Moreover, since the optimal policy depends only on the current observation $S_t$, all previous measurements can be discarded, thereby saving sensor memory.

\begin{assumption}\label{assumption:age-functions}
    The age functions, source pattern, and channel condition satisfy the following convergence conditions:
    \begin{align}
        \lim_{\delta\rightarrow\infty}\frac{g_{i,j}(\delta+1)}{g_{i,j}(\delta)}< \frac{1}{Q_{i, i}p_f}, ~\forall i\in\set{X}_{\rm{ap}}, i\neq j,\label{eq:existence-condition}
    \end{align}
    where $Q_{i, i}p_f$ represents the probability of remaining in error $(i,j)$ after each transmission attempt.
\end{assumption}

\begin{theorem}\label{theorem:existence}
    Suppose Assumption~\ref{assumption:age-functions} holds. Then there exists an optimal stationary deterministic policy $\pi^*$ such that for every $s\in\set{S}$, $a = \pi^*(s)$ solves the following Bellman equation:
    \begin{align}
        \set{L}^* + h(s) = \min_{a}\{l(s, a)+\sum\nolimits_{s^\prime}P_{s,s^\prime}(a)h(s^\prime)\},\label{eq:bellman-equation}
    \end{align}
    where $h$ is a bounded function, $\set{L}^*$ is the minimal average cost independent of the initial state $s_1$.
\end{theorem}
\begin{proof}
    See Appendix~\ref{proof:existence}.
\end{proof}

\begin{remark}
    The conditions in Assumption~\ref{assumption:age-functions} are relatively relaxed if the source evolves rapidly and the channel condition is good, i.e., when $Q_{i, i}$ and $p_f$ are small. The intuition behind this is that frequent state changes in the source lead to more frequent error variations (age drops), and good channel conditions increase the likelihood of successful transmissions. Note that in the remote estimation of linear Gaussian processes, the system is easier to stabilize if the source evolves slowly\cite{Dey-TAC-2017, WuShuang-TCNS-2020}. This notable contradiction arises because the error covariance of linear systems increases monotonically with time and only drops when a new measurement is received. From this perspective, the Markovian nature helps resist lasting impact.
\end{remark}

\begin{remark}
    The convergence condition depends on the source pattern and how we define the ``age". If the age is defined as the time elapsed since the last data reception (i.e., AoI), then the convergence condition is\footnote{The proofs of the convergence conditions for AoI and AoII are similar as in Theorem~\ref{theorem:existence} and are thus omitted for brevity. Notice that only one convergence condition is imposed for AoI and AoII since they treat different errors equally.}
    \begin{align}
        \lim_{\delta\rightarrow\infty}\frac{g(\delta+1)}{g(\delta)}< \frac{1}{p_f}.\label{eq:condition-AoI}
    \end{align}
    We note that \eqref{eq:condition-AoI} is independent of the source statistics. This implies that AoI is inefficient for Markovian sources.
    
    If AoII is considered, the condition becomes
    \begin{align}
        \lim_{\delta\rightarrow\infty}\frac{g(\delta+1)}{g(\delta)}< \frac{1}{1 - \min_{i\neq j}\{Q_{i, j}p_f + Q_{i, i}p_s\}}.
    \end{align}
    Given that
    \begin{align}
        \frac{1}{\max_i\{Q_{i, i}\}p_f} \hspace{-0.2em} > \hspace{-0.2em} \frac{1}{1 - \min_{i\neq j}\{Q_{i, j}p_f + Q_{i, i}p_s\}} \hspace{-0.2em} > \hspace{-0.2em} \frac{1}{p_f},
    \end{align}
    it follows that the more information we utilize in defining the ``age", the more relaxed the restrictions on age functions will be.
\end{remark}

As a result, we provide several age functions and derive their convergence conditions. It reveals that an optimal policy trivially exists when the age functions are bounded, linear, or logarithmic; it also holds for exponential functions with rate constraints. Similar analyses can be applied to other age metrics, e.g., \cite{George-GCW-2019, Maatouk-TON-AoII-2020, salimnejad2024version, luo2024MobiHoc}.

\begin{corollary}\label{corollary:age-functions}
    The assertion of Theorem~\ref{theorem:existence} holds if:
    \begin{itemize}
        \item[i.] $g_{i,j}(\delta)$ is an upper bounded clipping function, i.e.,
    \begin{align}
        g_{i,j}(\delta) = \begin{cases}
            g_{i,j}(\delta), &\text{if}~\delta < \Delta_{\max},\\
            g_{i,j}(\Delta_{\max}), &\text{if}~\delta \geq \Delta_{\max},
        \end{cases}\notag
    \end{align}
    where $\Delta_{\max}$ is a finite constant. 
    \item[ii.] $g_{i,j}(\delta) = \alpha \delta + \beta$ is linear, where $\alpha \geq 0$.
    \item[iii.] $g_{i,j}(\delta) = \log(\alpha \delta) + \beta$ is logarithmic, where $\alpha\geq 0$.
    \item[iv.] $g_{i,j}(\delta) = \alpha^{\beta \delta} + \zeta$ is an exponential function, where the base and power satisfy $\alpha^\beta < 1/ (Q_{i, i}p_f)$. For example, if $g_{i,j}(\delta) = e^{\beta \delta} + \zeta$, we impose $0\leq \beta < -\ln{(Q_{i,i}p_f)}$; if $g_{i,j}(\delta) = \alpha^{\delta} + \zeta$, we impose $1 \leq \alpha < 1/ (Q_{i,i}p_f)$.
    \end{itemize}
\end{corollary}

\subsection{Structure of an Optimal Policy}
Next, we show that the optimal policy has a switching structure. This result significantly remedies the ``curse of memory" and the ``curse of dimensionality." Moreover, under Assumption~\ref{assumption:source}, we show that the optimal policy degenerates into a threshold policy, i.e., featuring identical thresholds.

\begin{theorem}\label{theorem:structure}
    The optimal policy, if it exists, exhibits a switching structure. That is, for any given estimation error $(i, j)$ with an age $\delta \geq 1$, the sensor initiates a transmission only when $\delta$ exceeds a fixed threshold $\tau_{i,j}^* \geq 1$. Formally, 
    \begin{align}
        \pi^*(s) = \begin{cases} 
            1, & {\rm{if}}~\delta \geq \tau_{i,j}^*,i\neq j,\\ 
            0, & {\rm{otherwise}},
        \end{cases}\label{eq:switching-policy}
    \end{align}
    where $\tau_{i,j}^* = 1$ means always transmitting in estimation error $(i, j)$, whereas $\tau_{i,j}^* \rightarrow \infty$ means no transmission. 
\end{theorem}
\begin{proof}
    See Appendix~\ref{proof:structure}.
\end{proof}

Combining Lemma~\ref{lemma:sufficient-statistic} and Theorem~\ref{theorem:structure} gives the following:
\begin{corollary}\label{corollary:threshold-policy}
    Under Assumption~\ref{assumption:source}, the optimal policy degenerates to a simple threshold policy, i.e., $\tau_{i,j}^* = \tau^*, \forall i\neq j$. 
\end{corollary}
\begin{proof}
    See Appendix~\ref{proof:threshold-policy}.
\end{proof}

\begin{remark}
    The nice property of the switching policy \eqref{eq:switching-policy} is its simplicity in implementation and computation: i) to embed the optimal policy in the sensor memory, one only needs to store at most $M(M-1)$ threshold values instead of all possible state-action pairs; ii) to implement the policy in a real-time sensor scheduler, trigger a transmission only when the age of error exceeds the threshold and remain silent otherwise; iii) to compute the optimal policy, it suffices to search for a small number of optimal thresholds as opposed to solving a high-dimensional dynamic programming recursion.
\end{remark}

\begin{remark}
    Theorem~\ref{theorem:structure} answers the fundamental question of ``what and when to transmit". Distortion-optimal policies specify a deterministic mapping from estimation error to transmission decision\cite{luo2025TCOM}, guiding us on ``whether to transmit when a certain error occurs". By exploiting the lasting impact, our approach further determines the optimal timing to initiate a transmission for each estimation error. Therefore, existing results on distortion can be viewed as special cases of our framework where the thresholds are restricted to $\tau_{i,j}^* = 1$ (always-transmit) or $\tau_{i,j}^* \rightarrow \infty$ (non-transmit). An important takeaway is that the more semantic attributes we utilize, the fewer transmissions are needed.
\end{remark}

\vspace{-0.2in}
\subsection{Problem Approximation and Asymptotic Optimality}
We have shown several convenient properties of the optimal policy. However, it is challenging to analytically calculate the switching policy for general source patterns and age functions because: i) Computing the stationary distribution $\mu^\pi$ of the underlying system induced by a switching policy $\pi$ is a non-trivial task. ii) The average cost $\set{L}(\pi) = \sum_{s\in\set{S}}\mu^\pi(s)l(s,\pi(s))$ may not have closed-form expressions for general cases. iii) Even with such results at hand, $\set{L}(\pi)$ is a high dimensional and non-convex function, which does not translate into analytical expressions for the optimal thresholds\cite{luo2024MobiHoc}.

Thus, we resort to numerical methods to compute the optimal policy. Although Theorem~\ref{theorem:structure} considerably reduces the policy searching space, finding the optimal thresholds is still challenging as we cannot iterate over infinitely many states. To make the problem numerically tractable, we truncate the age process and propose a finite-state approximate MDP. The truncated AoCE is defined as
\begin{align}
    \Delta_{t}(N) \triangleq \min\{\Delta_t, N\}\label{eq:truncated-aoce},
\end{align}
where $\Delta_t$ is the original age process defined in \eqref{eq:aoce}, $N>0$ is the truncation size. In this way, the age associated with each estimation error $(i, j)$ is confined within
\begin{align}
    \set{S}_{i,j}^N  = \begin{cases}
        \{(i, j, \delta): 1\leq \delta\leq N\}, &\text{if}~i\neq j, i\in\set{X}_\textrm{ap},\\
        \{(i, j, 1)\}, &\text{if}~i\neq j, i\in\set{X}_\textrm{p}.\\
    \end{cases}
\end{align}
The truncated state space is $\set{S}^N = \set{S}_0 \bigcup_{i\neq j}\set{S}^N_{i,j}$.

The state space truncation, however, inevitably changes the stationary distribution of the induced Markov chain, yielding inconsistent system performance. Consequently, an optimal policy for the truncated MDP\footnote{One can easily show that the optimal policy for the truncated MDP is a switching type as well, according to the same proof as Theorem~\ref{theorem:structure}.} may be suboptimal for the original problem. Generally, there is no guarantee that the truncated MDP will converge to the original one as $N$ approaches infinity\cite{sennott1998stochastic}. Thus, we will be concerned with the performance loss caused by state space truncation. Fortunately, the following result establishes the asymptotic optimality of the truncated MDP, and \emph{we shall feel safe to truncate the AoCE with an appropriately chosen $N$.}

\begin{theorem}\label{therem:asymptotical-optimality}
    Suppose Assumption~\ref{assumption:age-functions} holds. Let $\set{L}^*(N)$ denote the minimal cost of the truncated MDP. Then, $\set{L}^*(N)$ converges to the optimal value $\set{L}^*$, i.e., $\lim_{N\rightarrow\infty}\set{L}^*(N) \rightarrow \set{L}^*$, at an exponential rate of $\max_{i\in\set{X}_{\rm{ap}}}\{(Q_{i,i}p_f)^N\}$.
\end{theorem}
\begin{proof}
    See Appendix~\ref{proof:asymptotical-optimality}.
\end{proof}

\vspace{-0.2in}
\subsection{Structured Policy Iteration Algorithm}
Classical unstructured policy iteration methods can be applied to solve the Bellman equation~\eqref{eq:bellman-equation} and obtain an optimal stationary deterministic policy\cite{puterman1994markov, ChenYutao-TON-AoII-2024}. Let $\Pi^\textrm{SD}$ and $\Pi^*$ denote the space of stationary deterministic policies and the space of switching policies, respectively, where $\Pi^*\subset \Pi^\textrm{SD}$. Assume the AoCE comprises $m$ ($m = |\set{X}_\textrm{ap}|(M-1) \leq M(M-1)$) sub age processes (see Remark~\ref{remark:age-process}), each with a truncation size of $N$ ($m\ll N$). Recall that there are $|\set{S}^N| = M^2 - m + mN$ states and $|\set{A}| = 2$ possible actions. The numbers of possible policies in $\Pi^*$ and $\Pi^\textrm{SD}$ are given by
\begin{align}
    |\Pi^*| = N^{m},~|\Pi^\textrm{SD}|  = |\set{A}|^{|\set{S}^N|} = 2^{M^2 - m + mN}\approx2^{mN}.\notag
\end{align}
Taking a logarithm of these numbers yields
\begin{align}
    \log_2(|\Pi^*|) = m \log_2(N) \ll \log_2(|\Pi^\textrm{SD}|) = mN.\notag
\end{align}
Therefore, the reduction in the size of the structured policy space is quite significant compared to the set of all deterministic policies. Moreover, unstructured iteration methods evaluate all possible policies in a stochastic manner, which makes the convergence rate even slower.

Next, we propose a \emph{structured policy iteration} (SPI) algorithm that exploits the switching structure established in Theorem~\ref{theorem:structure}. The algorithm proceeds as follows:
\begin{enumerate}
    \item \textit{Initialization:} Arbitrarily select an initial policy $\pi^n$, a reference state $s_\textrm{ref}$, and set $n=0$. Choose a truncation size $N$ such that $(Q_{i,i}p_f)^N < \epsilon$ for all $i\in\set{X}_\textrm{ap}$, where $\epsilon>0$ is an arbitrarily small constant.
    \item \textit{Policy Evaluation:} Find a scalar $\set{L}^n$ and a vector $h^n$ by solving
    \begin{align}
        \set{L}^n + h^n(s) = l(s, \pi^n(s)) + \hspace{-0.2em}\sum_{s^\prime\in\set{S}^N}P_{s,s^\prime}(\pi^n(s))h^n(s^\prime)\notag
    \end{align}
    for all $s\in\set{S}^N$ such that $h^n(s_\textrm{ref}) = 0$.
    \item \textit{Policy Improvement:} For each error $(i, j)$, update $\pi^{n+1}(i, j, \delta)$ in the increasing order of the AoCE $\delta$:
    \begin{enumerate}
        \item[a.] Initialize $s = (i, j, \delta)$ with $\delta = 1$. 
        \item[b.] Update the policy as
        \begin{align}
            \pi^{n+1}(s)
           \hspace{-0.2em} = \hspace{-0.2em} \argmin_{a\in\set{A}} \bigg[l(s, a) + \hspace{-0.5em}\sum_{s^\prime\in\set{S}^N}\hspace{-0.5em}P_{s,s^\prime}(a)h^n(s^\prime)\bigg].\notag
        \end{align}
        \item[c.] If $\pi^{n+1}(s) = 1$, the optimal action for all subsequent states $s^\prime = (i,j,\delta^\prime)$ with $\delta^\prime>\delta$ is to transmit without further computation. Thus, set $\pi^{n+1}(s^\prime) = 1$ for all $\delta^\prime\geq \delta$ and proceed to step 3(a) with an unvisited error. Otherwise, increment $\delta$ by setting $s = (i, j, \delta + 1)$ and return to step 3(b).
    \end{enumerate}
    \item \textit{Stopping Criterion:} If $\pi^{n+1} = \pi^n$, the algorithm terminates with $\set{L}^* = \set{L}^n$ and $\pi^* = \pi^n$; otherwise increase $n=n+1$ and return to step~2.
\end{enumerate}

The time complexity of SPI depends on the switching curve of the optimal policy rather than the full state space. The total number of operations in step 3 scales as $\sum_{i\neq j}\tau_{i,j}^*$, which is bounded by $\mathcal{O}(M^2N)$. This is significantly smaller than the $\mathcal{O}(|\mathcal{S}^N|^2|\mathcal{A}|)$ complexity of classical policy iteration. For instance, when $M = N= 10$, SPI requires at most $M^2N = 10^3$ operations, while classical method involves $|\mathcal{S}^N|^2|\mathcal{A}|\approx 2\times 10^6$ operations.

\begin{remark}
    We note that, while it is feasible to obtain an asymptotically optimal policy for general source patterns and age functions through age process truncation, Assumption~\ref{assumption:age-functions} remains a necessary prerequisite for the theorems established earlier. For instance, consider a digital twin that maintains a virtual representation of a physical system. It is essential to ensure that the policy computed in the virtual twin yields bounded average costs when applied to the real system.
\end{remark}

\vspace{-0.1in}
\begin{table}[ht]
\centering
\scriptsize
\setlength{\tabcolsep}{3pt}
\caption{Optimal thresholds for different communication costs $\lambda$ with $p_s = 0.9, N = 20$.}
\vspace{-0.1in}
\label{table:optimal-policy-lambda}
\renewcommand{\arraystretch}{0.8}
\begin{tabular}{cccc|c|c|c}
    \toprule
    \multirow{2.5}{*}{$\lambda$} & \multicolumn{3}{c|}{Switching policy} & Distortion & AoI & AoII \\
    \cmidrule(lr){2-7} 
    & Missed alarms & False alarms & Normal errors & All errors & All states & All errors \\
    \midrule
    0   & 1 & 1       & 1       & 1       & 1       & 1 \\
    1   & 1 & 1       & 1       & 1       & 2       & 1 \\
    2   & 1 & 1       & $\infty$& $\infty$& 2       & 1 \\
    3   & 2 & 3       & $\infty$& $\infty$& 3       & 1 \\
    4   & 3 & 11      & $\infty$& $\infty$& 3       & 1 \\
    5   & 3 & $\infty$& $\infty$& $\infty$& 3       & 2 \\
    6   & 3 & $\infty$& $\infty$& $\infty$& 4       & 2 \\
    7   & 3 & $\infty$& $\infty$& $\infty$& 4       & 3 \\
    \bottomrule
\end{tabular}
\end{table}
\vspace{-0.2in}
\begin{table}[ht]
\centering
\scriptsize
\setlength{\tabcolsep}{3pt}
\caption{Optimal thresholds for different channel conditions $p_s$ with $\lambda = 3, N = 20$.}
\vspace{-0.1in}
\label{table:optimal-policy-ps}
\renewcommand{\arraystretch}{0.8}
\begin{tabular}{cccc|c|c|c}
    \toprule
    \multirow{2.5}{*}{$p_s$} & \multicolumn{3}{c|}{Switching policy} & Distortion & AoI & AoII \\
    \cmidrule(lr){2-7} 
    & Missed alarms & False alarms & Normal errors & All errors & All states & All errors \\
    \midrule
    0.1   & 1    & $\infty$ & $\infty$   & $\infty$ & 4  & 5\\  
    0.2   & 1    & $\infty$ & $\infty$   & $\infty$ & 3  & 3\\ 
    0.3   & 1    & $\infty$ & $\infty$   & $\infty$ & 3  & 2\\ 
    0.4   & 2    & $\infty$ & $\infty$   & $\infty$ & 3  & 2\\ 
    0.5   & 2    & $\infty$ & $\infty$   & $\infty$ & 3  & 1\\ 
    0.6   & 3    & $\infty$ & $\infty$   & $\infty$ & 3  & 1\\ 
    0.7   & 3    & 10       & $\infty$   & $\infty$ & 3  & 1\\ 
    0.8   & 3    & 5        & $\infty$   & $\infty$ & 3  & 1\\ 
    0.9   & 2    & 3        & $\infty$   & $\infty$ & 3  & 1\\ 
    1.0   & 2    & 2        & $\infty$   & $\infty$ & 3  & 1\\ 
    \bottomrule
\end{tabular}
\end{table}

\vspace{-0.2in}
\section{Numerical Results}\label{sec:numerical-results}
In this section, we provide numerical simulations to illustrate our results on the optimal switching policy. 

\textbf{Behavioral analysis:} In the first example, we show the impact of significance-aware non-linear aging on the switching curve. To isolate the effect of other factors, we consider a symmetric source of the form in~\eqref{eq:symmetric-source} with the following parameters: $M = 4, p = 0.1, \bar{p} = 0.7$, and $D_{i,j}=1, \forall i\neq j$. We assign exponential age functions to missed alarms and logarithmic age functions to false alarms, i.e.,
\begin{align}
    g_{i,j}(\delta) = \begin{cases}
        e^{0.3\delta}, &\textrm{if}~i = 1, j\neq 1,\\
        \log(\delta) + 1, &\textrm{if}~i \neq 1, j=1,\\
        1, &\textrm{otherwise.}
    \end{cases}\notag
\end{align}
This assignment satisfies the convergence condition \eqref{eq:existence-condition}. By Theorem~\ref{theorem:structure}, there exists an optimal switching policy.

The optimal thresholds for different communication costs $\lambda$, with a fixed success probability $p_s$, are shown in Table~\ref{table:optimal-policy-lambda}. Similarly, the thresholds for different $p_s$, with a fixed $\lambda$, are presented in Table~\ref{table:optimal-policy-ps}. We use the SPI to find the optimal switching policy. The AoI- and AoII-optimal policies are obtained by classical policy iteration\cite{ChenYutao-TON-AoII-2024}. When communication is costly or the channel condition is relatively poor, the optimal switching policy is to transmit less frequently (or never transmit) in false alarms and normal errors while consistently prioritizing missed alarms. In contrast, the distortion-optimal policy either initiates transmissions for all errors (i.e., $\tau_{i,j} = 1, \forall i\neq j$) when communication is inexpensive ($\lambda\leq 1$) or remains silent otherwise (i.e., $\tau_{i,j} \rightarrow \infty, \forall i\neq j$). Due to source symmetry, the AoI- and AoII-optimal policies specify a single threshold for all errors\cite{Maatouk-TON-AoII-2020}. The AoI metric shows obvious disadvantages as it completely ignores the source pattern and initiates transmissions even in synced states. The AoII metric is also inefficient in our problem since it treats all errors equally, resulting in excessive transmissions for less critical errors. From Table~\ref{table:optimal-policy-ps}, we observe that the optimal threshold for missed alarms does not always increase as the channel quality $p_s$ degrades. The reason is that the sensor must schedule missed alarms more frequently to counteract the exponential growth of the lasting cost caused by the high probability of transmission failures. \emph{Therefore, the significance-aware AoCE offers more informed decisions and extends the current understanding of distortion and information aging.}

\begin{table*}[ht]
\centering
\scriptsize
\setlength{\tabcolsep}{5.5pt}
\caption{Performance comparison of different policies with $p_s = 0.9, N = 20$.}
\vspace{-0.1in}
\label{table:performance-comparison}
\renewcommand{\arraystretch}{0.8}
\begin{tabular}{ccccccccccc}
    \toprule
    \multirow{2.5}{*}{$\lambda$} & \multicolumn{9}{c}{Achievable minimal average costs} \\
    \cmidrule(lr){2-10}
    & Randomized & Periodic & Reactive & Error-triggered & Threshold & Distortion & AoI & AoII & \textbf{Switching} \\
    \midrule
    0  & \textbf{0.31} & \textbf{0.31} & 0.43 & \textbf{0.31}  & \textbf{0.31} &\textbf{0.31}&\textbf{0.31}&\textbf{0.31}& \textbf{0.31} \\
    1  & 0.93  &0.85   &0.72 &0.62 & 0.62 &\textbf{0.59}&0.84  & 0.62  & \textbf{0.59} \\
    2  & 1.46  & 1.05  &0.98 &0.91 & 0.88 &1.32&1.14  &0.91   & \textbf{0.73} \\
    3  & 1.20  & 1.20  &1.27 &1.21 & 0.97 &1.34&1.42  &1.14   & \textbf{0.80} \\
    4  & 1.20  & 1.20  &1.55 &1.50 & 1.01 &1.34&1.65  &1.34   & \textbf{0.85} \\
    5  & 1.20  & 1.20  &1.83 &1.80 & 1.04 &1.34&1.84  &1.34   & \textbf{0.88} \\
    \bottomrule
\end{tabular}
\end{table*}

\textbf{Performance comparisons:} We now consider a general asymmetric source with 
\begin{align}
    Q = \begin{bmatrix}
        0.7 & 0.1 & 0.1 & 0.1\\
        0.05& 0.7 & 0.15& 0.1\\
        0.1 & 0.1 & 0.6 & 0.2\\
        0.05& 0.1 & 0.05& 0.8
    \end{bmatrix}.\notag
\end{align}
The other parameters are the same as in the first example. For comparison purposes, we consider the following rule-based policies\cite{SunYin-TIT-2020, Nikos-CAE-2021}:
\begin{itemize}
    \item \textit{Randomized}: The sensor transmits at every slot $t$ with a fixed probability $p_\alpha\leq 1$. 
    \item \textit{Periodic}: The sensor transmits every $d_\textrm{th} \geq 1$ slots, and remains salient otherwise.
    \item \textit{Reactive}: A new transmission is triggered only on state changes, i.e., $X_t \neq X_{t-1}$.
    \item \textit{Error-triggered}: A new transmission is triggered whenever an error occurs, i.e., $X_t \neq \hat{X}_{t-1}$.
    \item \textit{Threshold}: A new transmission is triggered once the AoCE exceeds a threshold $\delta_\textrm{th}\geq 1$, irrespective of the instantaneous estimation error. The optimal threshold policy is global optimal when the source is symmetric and all states are equally important (see~Corellary~\ref{corollary:threshold-policy}).
\end{itemize}

The first three policies are open-loop policies since they do not require feedback links. The optimal values $p_\alpha^*$, $d^*_\textrm{th}$, and $\delta^*_\textrm{th}$ for the randomized, periodic, and threshold policies are obtained by brute force search. When testing policy performance, we take the average cost over $10^7$ samples. A comparison of these policies is summarized in Table~\ref{table:performance-comparison}. It is observed that when communication is cost-free ($\lambda = 0$), most policies adopt the always-transmit strategy. As $\lambda$ increases, the classic rule-, distortion- and age-based policies become increasingly inefficient. This is because they disregard data significance and transmit measurements in less critical errors. Specifically, they make decisions based on insufficient statistics. For example, the distortion policy relies on $(X_t, \hat{X_t})$, while the threshold policy utilizes only the AoCE $\Delta_t$.

The switching curve also depends on the source pattern. For instance, the optimal thresholds for $\lambda=3$ are obtained as
\begin{align}
    [\tau^*_{i,j}]_{4\times 4} = \begin{bmatrix}
        \infty & 2 & 2 & 3\\
        3 & \infty & \infty & \infty\\
        \infty & \infty & \infty & \infty\\
        1 & 1 & 1 & \infty
    \end{bmatrix}.\label{eq:thresholds}\notag
\end{align}
We observe that the thresholds for missed alarms (first row) are larger than those for normal errors (last row). This occurs because $Q_{4,4}$ is large, and the system is likely to remain in these normal errors for an extended period if no measurements are received. So it is beneficial to schedule frequently these normal errors, even though their age penalties are relatively small. In outline, the switching policy determines the optimal timing for data transmission based on the significance and dynamic evolution of estimation errors. \emph{This highlights the effectiveness of exploiting significance and timing aspects of information in such systems.}

\section{Conclusion}\label{sec:conclusion}
In this paper, we have investigated the semantics-aware remote estimation of prioritized Markov chains through the significance-aware AoCE. We first give sufficient conditions for the existence of an optimal policy. We prove that a switching policy is optimal and develop a structure-aware algorithm to find the optimal thresholds with reduced computation. Numerical comparisons show that the optimal policy can be much better than existing rule-, distortion- and age-based policies. The results in this paper generalize recent research on distortion and information aging.

\appendices
\section{Proof of Lemma~\ref{lemma:sufficient-statistic}}\label{proof:sufficient-statistic}
The AoCE $\Delta_t$ is a sufficient statistic if it satisfies~\cite{mahajan2016decentralized}: 
\begin{enumerate}
    \item the controlled Markov property
    \begin{align}
        \Pr[\Delta_{t+1}|I_t, A_t] = \Pr[\Delta_{t+1}|\Delta_t, A_t],
    \end{align}
    \item the absorbing property
        \begin{align}
           \Pr[c(I_t)|I_t, A_t] = \Pr[c(\Delta_t)|\Delta_t, A_t].
        \end{align}
\end{enumerate}

Recall that the AoCE~\eqref{eq:aoce} depends on the source pattern $Q$ and transmission decision $A_t$. Therefore, knowing the current age $\Delta_t$ is inadequate for inferring the next age $\Delta_{t+1}$. When the transition matrix $Q$ is of the special form described in \eqref{eq:symmetric-source}, the state change probabilities become independent of the current state $X_t$ and the target state $X_{t+1}$. Hence, it makes no difference which error the system encounters.  What matters is whether the system is in an erroneous state and how long it has stayed there. This information is fully captured by the age $\Delta_t$. From the system dynamics in \eqref{eq:prob-error-0}-\eqref{eq:prob-synced-states}, the transition probabilities of $\Delta_t$ are given by
\begin{align}
    \Pr[\delta^\prime|\delta, 0] &= \begin{cases}
        \bar{p}, &\text{if}~\delta = 0, \delta^\prime=0,\\
        1-\bar{p}, &\text{if}~\delta = 0, \delta^\prime=1,\\
        p, &\text{if}~\delta >0, \delta^\prime=0,\\
        \bar{p}, &\text{if}~\delta >0, \delta^\prime=\delta+1,\\
        1-p-\bar{p}, &\text{if}~\delta >0, \delta^\prime=1,
    \end{cases} \label{eq:age-symmetric-0} \\
    \Pr[\delta^\prime|\delta, 1] &= \begin{cases}
        \bar{p}, &\text{if}~\delta = 0, \delta^\prime=0,\\
        1-\bar{p}, &\text{if}~\delta = 0, \delta^\prime=1,\\
        \bar{p}p_s + pp_f, &\text{if}~\delta >0, \delta^\prime=0,\\
        \bar{p}p_f, &\text{if}~\delta >0, \delta^\prime=\delta+1,\\
        1-\bar{p} - pp_f, &\text{if}~\delta >0, \delta^\prime=1.
    \end{cases}\label{eq:age-symmetric-1}
\end{align}
Hence, the controlled Markov property follows. Further, when the source is non-prioritized, the cost function becomes $c(I_t) = Dg(\Delta_t)$, which satisfies the absorbing property.

\section{Proof of Lemma~\ref{lemma:recurrence}}\label{proof:recurrence}
Since the AoCE $\Delta_t$ is uniquely determined by the realization of $Z_t = (X_t, \hat{X}_t)$, it suffices to show that the subsystem $Z_t$, when controlled by the always-transmit policy, is positive recurrent. From \eqref{eq:prob-error-0}-\eqref{eq:prob-synced-states}, the state transition matrix of $Z_t$ under the always-transmit policy are obtained as
\begin{align}
    R = (R_{z,z^\prime}, z, z^\prime\in\set{Z}),
\end{align}
where $\set{Z} = \set{X}\times\set{X}$ and 
\begin{align}
    R_{z,z^\prime} &= \Pr[Z_{t+1}=z^\prime|Z_t = z, A_t = 1]\notag\\
    &=\begin{cases}
        Q_{i,k}p_s, &\text{if}~z=(i,j),z^\prime=(k,i),\\
        Q_{i,k}p_f, &\text{if}~z=(i,j),z^\prime=(k,j),\\
        Q_{i,k}, &\text{if}~z=(i,i),z^\prime=(k,i),\\
        0, &\text{otherwise.}
    \end{cases}
\end{align}

Given that $Z_t$ is a finite-state Markov chain, the positive recurrence property is verified if $R$ is irreducible\cite{gallager1997discrete}; that is, every pair of states $(z, z^\prime)$ communicates with each other. Let $z = (i,j)$ and $z^\prime=(m, n)$ be two distinct states in $\set{Z}$. Since $Q$ is irreducible, the source can jump from state $i$ to state $n$ in a finite number of steps. Suppose this transition takes $t_1$ steps and all transmission attempts fail. This probability is 
\begin{align}
    \Pr[(i,j)\rightarrow(n,j)] = Q_{i, n}^{t_1}(p_f)^{t_1}>0.\label{eq:prob-a}
\end{align}
Suppose that in the subsequent time slot, a new measurement is received w.p. $p_s$ and the estimate becomes $n$. After that, the source moves from state $n$ to state $m$ in $t_2$ steps while all transmission attempts fail. Thus, we have
\begin{align}
    \Pr[(n,j)\rightarrow(m,n)] = Q_{n, m}^{t_2+1}(p_f)^{t_2} p_s >0.\label{eq:prob-b}
\end{align}
From \eqref{eq:prob-a} and \eqref{eq:prob-b}, the probability of the system transitioning from state $z$ to $z^\prime$ along the above sample path is
\begin{align}
    \Pr[z\rightarrow z^\prime] = Q_{i, n}^{t_1}Q_{n, m}^{t_2+1}(p_f)^{t_1+t_2} p_s >0.
\end{align}

Similarly, we can show that there exists a sample path from $z^\prime$ to $z$ along which the probability $\Pr[z\leftarrow z^\prime]$ is positive. Therefore, the chain $Z_t$ under the always-transmit policy is irreducible and positive recurrent. It follows that the expected first passage time from every state $z$ to $z^\prime$ is finite. 

\section{Proof of Theorem~\ref{theorem:existence}}\label{proof:existence}
Sufficient conditions ensuring the existence of an optimal policy for average-cost MDPs with countably infinite state spaces are presented in~\cite[Corollary~7.5.10]{sennott1998stochastic}. We need to verify the following conditions: (C1) Given any positive constant $U$, the set $\set{S}_U\triangleq\{s\in\set{S}:l(s, a)\leq U~\textrm{for some}~ a\}$ is finite. (C2) There exists a distinguished state $z\in\set{S}$ and a $z$-standard policy (see Definition~\ref{def:standard policy}). Based on the cost function~\eqref{eq:total-cost}, C1 holds trivially.

\begin{definition}\label{def:standard policy}
    A stationary policy is called standard if the underlying system induced by the policy forms an irreducible and positive recurrent Markov chain. Further, if there exists a state $z$ such that the expected first passage time and the expected first passage cost from any state $s$ to $z$ are finite, the policy is called $z$-standard.
\end{definition}

To verify C2, let $\pi$ be the always-transmit policy. By Lemma~\ref{lemma:recurrence}, there exists a state $z\in\set{S}$, say $z = (0, 0, 0)$, such that $\pi$ is standard and the expected first passage time from any state $s$ to $z$ is finite. It remains to show that the expected first passage cost from any state $s$ to $z$ is also finite. Since the AoCE drops upon error variations, it suffices to ensure that the expected cost incurred in each estimation error $(i,j)$ is finite. To establish this, we first define a fictitious ``zero" state.

\begin{definition}\label{def:fictitious-zero}
    For each estimation error $(i, j), i\in\set{X}_{\rm{ap}}, i\neq j$, we define a fictitious ``zero" state that consists of all the synced states and other estimation errors with unit age, i.e.,
    \begin{align}
        (i, j, 0) \triangleq \set{S}_0 \cup\{(m, n, 1): \forall m\neq n, (m, n) \neq (i, j)\}.
    \end{align}
\end{definition}

\begin{figure}[ht]
    \centering
    \scalebox{0.88}{\begin{tikzpicture}[node distance=2.35cm, on grid, >=stealth, state/.style={circle, draw, minimum size=0.9cm, inner sep=0pt, font=\small, line width=0.3mm}]
    
    \node[state] (s1) {$(i, j, 0)$};
    \node[state, right=of s1] (s2) {$(i, j, 1)$};
    \node[state, right=of s2] (sD0) {$(i, j, 2)$};
    \node[state, right=of sD0] (sD1) {$(i, j, \delta)$};
    \node[right=2cm of sD1] (dots) {};

    \draw[->, line width=0.3mm] (s1) -- (s2);
    \draw[->, line width=0.3mm] (s2) -- (sD0) node[midway, above,yshift=-2pt] {\small $Q_{i,i}p_f$};
    \draw[->, dashed, line width=0.3mm] (sD0) -- (sD1) node[midway, above,yshift=-2pt] {\small $Q_{i,i}p_f$};
    \draw[->, line width=0.3mm] (sD1) -- (dots) node[midway, above,yshift=-2pt]{\small $Q_{i,i}p_f$};

    \draw[->, bend right=30, line width=0.3mm] (s2) to node[pos=0.3,  above right, sloped, yshift=-3pt] {} (s1);
    \draw[->, bend right=30, line width=0.3mm] (sD0) to node[pos=0.3,  above right, sloped, yshift=-3pt] {} (s1);
    \draw[->, bend right=30, line width=0.3mm] (sD1) to node[pos=0.3,  above right, sloped, yshift=-3pt] {\small $1 - Q_{i,i}p_f$} (s1);
    
    \draw[->, loop above, line width=0.3mm] (s1) to node[above, yshift=-2pt] {} (s1);
    
\end{tikzpicture}}
    \caption{The evolution of the AoCE on error $(i, j)$. Here, $Q_{i, i}p_f$ is the probability of remaining in this error after each transmission attempt, while $1-Q_{i, i}p_f$ is the probability of leaving this error.}
    \label{fig:theorem-1-proof}
\end{figure}

Fig.~\ref{fig:theorem-1-proof} depicts the evolution of the error holding time induced by the always-transmit policy. Assume the system is initialized in error $(i,j)$ with an age $\delta$. The probability of the system staying in error $(i, j)$ for $n$ consecutive time slots is
\begin{align}
    p_{i,j}(n) = (1-Q_{i,i}p_f)(Q_{i, i}p_f)^n. \label{eq:pn}
\end{align}
The first passage cost along this sample path is
\begin{align}
    M_{i,j}(n) = \sum_{t=0}^{n} 
    \big( 
      D_{i,j} g_{i,j}(\delta+t) + \lambda
    \big) + \tilde{D},\label{eq:Mn}
\end{align}
where $\tilde{D}$ is the cost incurred in the ``zero" state. Since $\tilde{D}$ is a finite constant, we omit it for brevity. 

From \eqref{eq:pn} and \eqref{eq:Mn}, the expected first passage cost is
\begin{align}
    \bar{M}_{i,j} &= \sum_{n = 0}^{\infty}p_{i,j}(n)M_{i,j}(n)\notag\\
    &= \sum_{n = 0}^{\infty}(1-Q_{i,i}p_f)(Q_{i, i}p_f)^n 
    \big( 
      D_{i,j} g_{i,j}(\delta+n) + \lambda \big)\notag\\
    &= \lambda + (1-Q_{i,i}p_f)D_{i,j} \sum_{n = 0}^{\infty}\xi_{i,j}(n),
\end{align}
where $\xi_{i,j}(n) = (Q_{i, i}p_f)^n g_{i,j}(\delta+n)$.

The proof is therefore reduced to showing that the summation $\sum_{n = 0}^{\infty}\xi_{i,j}(n)$ is finite. By the ratio test, the series $\{\xi_{i,j}(n)\}_{n\geq 0}$ converges if 
\begin{align}
    \lim_{n\rightarrow\infty}\frac{\xi_{i,j}(n+1)}{\xi_{i,j}(n)} < 1.
\end{align}
Rearranging the terms yields
\begin{align}
    \lim_{n\rightarrow\infty}\frac{g_{i,j}(\delta + n+1)}{g_{i,j}(\delta + n)} = \lim_{n^\prime \rightarrow\infty}\frac{g_{i,j}(n^\prime+1)}{g_{i,j}(n^\prime)}< \frac{1}{Q_{i,i}p_f}.
\end{align}
This concludes our proof.

\section{Proof of Theorem~\ref{theorem:structure}}\label{proof:structure}
Recall that transmitting in synced states has no effect on the estimation performance but incurs additional communication costs (see Eq.~\eqref{eq:prob-synced-states}). Thus, we have
\begin{align}
    \pi^*(s) = 0,~\textrm{if}~ i = j.\label{eq:case-0}
\end{align}
We are therefore only interested in erroneous states. The switching structure implies that the optimal policy $\pi^*(s)$ is monotonically non-decreasing in $\delta$ for any fixed estimation error $(i, j)$. By Theorem~\ref{theorem:existence}, the optimal policy satisfies $\pi^*(s) = \argmin_{a}f(s, a)$, where
\begin{align}
    f(s,a) = l(s, a) + \sum\nolimits_{s^\prime}P_{s,s^\prime}(a)h(s^\prime).\label{eq:f}
\end{align}

Establishing structural results requires the submodularity property of $f(s, a)$.

\begin{definition}\label{def:state-space-ordering}
    For each error $(i, j)$, where $i\in\set{X}_{\rm{ap}}$ and $i\neq j$, the states in the sub-space $\set{S}_{i, j}$ are ordered by the age $\delta$. That is, for any two states $s$ and $s^\prime$ in $\set{S}_{i, j}$, we define $s\leq s^\prime$ if their age satisfies $\delta \leq \delta^\prime$. 
\end{definition}

\begin{definition}\label{def:submularity} 
The function $f(s, a)$ is submodular on $\set{S}_{i, j}\times\set{A}$, if for all $s^\prime\geq s$ and $a^\prime\geq a$, it satisfies the inequality
\begin{align}
    f(s^\prime, a^\prime)+f(s, a) \leq f(s^\prime, a)+f(s, a^\prime).\label{eq:Q-factor-submodularity}
\end{align}
If the inequality is reversed, $f(s, a)$ is called supermodular.
\end{definition}

\begin{lemma}[{\hspace{-0.02em}\cite[Lemma~4.7.1]{puterman1994markov}}]\label{lemma:submodularity}
If $f(s, a)$ is submodular on $\set{S}_{i, j}\times\set{A}$, the optimal policy for error $(i, j)$, $\pi^*(s) = \argmin_{a\in\set{A}}f(s, a)$, is non-decreasing in $\delta$.
\end{lemma}

The proof is therefore reduced to verifying the inequality in \eqref{eq:Q-factor-submodularity}. Based on the system state transition probabilities derived in Section~\ref{sec:system-evolution}, the function $f(s, a)$ can be expressed as
\begin{align}
    f(s, 0) = c(s) &+ Q_{i, i} h(i, j, \delta + 1) + Q_{i, j} h(j, j, 0) \notag\\
    &+ \sum_{k\neq i,j}Q_{i, k} h(k, j, 1),\label{eq:f-1}\\
    f(s, 1) = c(s) & + \lambda 
    + Q_{i, i}p_s h(i, i, 0) 
    + \sum_{k_1\neq i}Q_{i, k_1}p_s h(k_1, i, 1)\notag\\
    &+ Q_{i,i}p_f h(i, j, \delta + 1)
     + Q_{i,j}p_f h(j, j, 0)\notag\\
    &+ \sum_{k_2\neq i,j}Q_{i, k_2}p_f h(k_2, j, 1).\label{eq:f-2}
\end{align}
Substituting \eqref{eq:f-1} and \eqref{eq:f-2} into \eqref{eq:Q-factor-submodularity} yields
\begin{align}
    f(s^\prime, 1) &+ f(s, 0) - f(s^\prime, 0) - f(s, 1) \notag\\
    &= - Q_{i,i}p_s \Big(
     h(i, j, \delta^\prime + 1) - h(i, j, \delta + 1)
    \Big).
\end{align}
We show in Theorem~\ref{theorem:monotone-h} that $h(i, j, \delta)$ is non-decreasing in $\delta$. We can thus conclude that $f(s, a)$ is submodular on $\set{S}_{i,j}\times\set{A}$. 

Therefore, for each error $(i,j)$, the optimal policy $\pi^*(s)$ is non-decreasing in $\delta$. That is, there exists a transmission threshold $\tau_{i,j}$, such that a transmission is initiated once the age $\delta$ exceeds this threshold. Formally, we write
\begin{align}
    \pi^*(s) = \begin{cases}
        1, &~\textrm{if}~\delta\geq \tau_{i,j}^*, i\neq j,\\
        0, &~\textrm{if}~ 1 \leq \delta < \tau_{i,j}^*, i\neq j.\\
    \end{cases}\label{eq:case-1}
\end{align}
Combining \eqref{eq:case-0} and \eqref{eq:case-1} yields the switching policy in \eqref{eq:switching-policy}.

\begin{theorem}\label{theorem:monotone-h}
    $h(s)$ is non-decreasing in $\delta$ for any fixed $(i, j)$.
\end{theorem}
\begin{proof}
    Theorem~\ref{theorem:existence} suggests that we can use the \emph{relative value iteration} (RVI)~\cite{puterman1994markov} to solve the Bellman equation~\eqref{eq:bellman-equation}. For each iteration $n\geq 1$, the RVI updates the value function using the following recursion:
\begin{align}
    f^{n}(s, a) &= l(s, a) + \sum\nolimits_{s^\prime}P_{s,s^\prime}(a) h^{n-1}(s^\prime),\label{eq:rvi-a}\\
    \tilde{h}^{n}(s) &= \min_{a} f^{n}(s, a),\label{eq:rvi-b}\\
    h^{n}(s) &= \tilde{h}^{n}(s) - \tilde{h}^{n}(s_\textrm{ref}),\label{eq:rvi-c}
\end{align}
where $s_\textrm{ref}\in \set{S}$ is an arbitrary reference state. The sequences $\{\tilde{h}^n(s)\}$ and $\{h^n(s)\}$ converge as $n\rightarrow\infty$\cite{puterman1994markov}. Moreover, $\lambda^* = \tilde{h}(s_\textrm{ref})$ and $h(s) = \tilde{h}(s) - \tilde{h}(s_\textrm{ref})$ are a solution to \eqref{eq:bellman-equation}. 

We prove the monotonicity of $h(s)$ by induction. Choose $h^0(i, j, \delta)$ to be non-decreasing in $\delta$, i.e.,
\begin{align}
    h^0(i,j, \delta) \leq h^0(i,j, \delta^\prime),~ \forall \delta \leq \delta^\prime.
\end{align}
Assume that $h^{n}(i,j, \delta)$ is non-decreasing in $\delta$. We need to show that $h^{n+1}(i,j, \delta)$ is still a non-decreasing function. By~\eqref{eq:rvi-a}-\eqref{eq:rvi-c}, we can write
\begin{align}
    h^{n+1}(s) = \min_{a}f^{n+1}(s,a) - \min_{a}f^{n+1}(s_\textrm{ref},a).
\end{align}
By \eqref{eq:f-1}-\eqref{eq:f-2} and the induction hypothesis, we have
\begin{align}
    f^{n+1}(s, 0) &= c(s) + Q_{i, i} h^n(i, j, \delta + 1) + Q_{i, j} h^n(j, j, 0) \notag\\
    &\qquad \qquad \qquad + \sum_{k\neq i,j}Q_{i, k} h^n(k, j, 1)\notag\\
    &\leq c(s) + Q_{i, i} h^n(i, j, \delta^\prime + 1) + Q_{i, j} h^n(j, j, 0) \notag\\
    &\qquad \qquad \qquad + \sum_{k\neq i,j}Q_{i, k} h^n(k, j, 1)\notag\\
    &= f^{n+1}(s^\prime, 0),\\
    f^{n+1}(s, 1) &= c(s) +\lambda 
    + Q_{i, i}p_s h^n(i, i, 0) \notag\\
    &+ \sum_{k_1\neq i}Q_{i, k_1}p_s h^n(k_1, i, 1)
    + Q_{i,i}p_f h^n(i, j, \delta + 1)\notag\\
    &+ Q_{i,j}p_f h^n(j, j, 0)
    + \sum_{k_2\neq i,j}Q_{i, k_2}p_f h^n(k_2, j, 1)\notag\\
    &\leq c(s) +\lambda 
    + Q_{i, i}p_s h^n(i, i, 0) \notag\\
    &+ \sum_{k_1\neq i}Q_{i, k_1}p_s h^n(k_1, i, 1)
    + Q_{i,i}p_f h^n(i, j, \delta^\prime + 1)\notag\\
    & + Q_{i,j}p_f h^n(j, j, 0)
    + \sum_{k_2\neq i,j}Q_{i, k_2}p_f h^n(k_2, j, 1)\notag\\
    &= f^{n+1}(s^\prime, 1).
\end{align}
It gives that
\begin{align}
    h^{n+1}(s) &= \min \hspace{-0.2em} \big\{
        f^{n+1}(s, 0), f^{n+1}(s, 1) 
    \big\} - \min_{a}f^{n+1}(s_\textrm{ref},a)\notag\\
    &\leq \min  \hspace{-0.2em} \big\{
        f^{n+1}(s^\prime, 0), f^{n+1}(s^\prime, 1) 
    \big\} \hspace{-0.2em} - \hspace{-0.2em} \min_{a}f^{n+1}(s_\textrm{ref},a)\notag\\
    &= \min_{a}f^{n+1}(s^\prime,a) - \min_{a}f^{n+1}(s_\textrm{ref},a)\notag\\
    & = h^{n+1}(s^\prime).
\end{align}
Therefore, $h^{n+1}(s)$ is a non-decreasing function. Since the monotonicity property is preserved as $n\rightarrow \infty$, we can conclude that $h(s)$ is non-decreasing in $\delta$ for any fixed $(i,j)$. This completes the inductive proof.
\end{proof}

\section{Proof of Corollary~\ref{corollary:threshold-policy}}\label{proof:threshold-policy}
Let Assumption~\ref{assumption:source} hold. By Lemma~\ref{lemma:sufficient-statistic}, the AoCE $\Delta_t$ is a sufficient statistic for the problem. Therefore, we need to show that the optimal policy to the MDP with system state $\Delta_t$ and transition probabilities \eqref{eq:age-symmetric-0}-\eqref{eq:age-symmetric-1} is a threshold policy. Similar to Theorem~\ref{theorem:structure}, the proof is reduced to showing that the optimal policy is non-decreasing in $\delta$. Thus, we will verify the submodularity of the function $f(s, a)$ defined in \eqref{eq:f}.

Based on the state transition probabilities \eqref{eq:age-symmetric-0}-\eqref{eq:age-symmetric-1}, $f(s,a)$ can be written as
\begin{align}
    f(s, 0) &= Dg(\delta) + p h(0) + \bar{p} h(\delta + 1) + (1 - p -\bar{p}) h(1),\label{eq:delta-0}\\
    f(s, 1) &= Dg(\delta) + \lambda + (\bar{p}p_s + p p_f) h(0) + \bar{p}p_f h(\delta + 1) \notag\\
    &\qquad \qquad + (1 - \bar{p} - pp_f) h(1), \label{eq:delta-1}
\end{align}
where $p$ and $\bar{p}$ are defined in \eqref{eq:symmetric-source}. Substituting \eqref{eq:delta-0} and \eqref{eq:delta-1} into the submodularity condition \eqref{eq:Q-factor-submodularity} yields
\begin{align}
    f(s^\prime, 1) + f(s, 0) &- f(s^\prime, 0) - f(s, 1) \notag\\
    &= - \bar{p} p_s \big(
     h(\delta^\prime + 1) - h(\delta + 1)
    \big).
\end{align}
We can show by induction that $h(\delta)$ is a non-decreasing function. It follows that $f(s, a)$ is submodular.

\section{Proof of Theorem~\ref{therem:asymptotical-optimality}}\label{proof:asymptotical-optimality}
We prove the result by two steps:
\begin{enumerate}
    \item For any given switching policy $\pi$, show that the cost of the truncated problem converges to the original one as $N\rightarrow\infty$, i.e., $\lim_{N\rightarrow\infty}\set{L}(\pi, N)\rightarrow\set{L}(\pi)$.
    \item Use the above result and the squeeze theorem to show that $\lim_{N\rightarrow\infty}\set{L}^*(N)\rightarrow\set{L}^*$.
\end{enumerate}

{\it{Step 1:}} Recall that the optimal policy for the truncated MDP is a switching type as well. Let $\pi$ denote a switching policy with all threshold values less than $N$, i.e., $\tau_{i,j} < N, \forall i\neq j$. Let $\mu$ and $\mu^\prime$ denote the stationary distribution of the original and truncated systems induced by policy $\pi$, respectively. Denote $\set{S}_{i, j}^\textrm{in}$ and $\set{S}_{i, j}^\textrm{out}$ as the sets of inner and outer states associated with error $(i, j)$, respectively, where
\begin{align}
    \set{S}_{i, j}^\textrm{in} &\triangleq \{(i,j,\delta) \in \set{S}_{i,j}: 1\leq \delta<N\},\\
    \set{S}_{i, j}^\textrm{out} &\triangleq \{(i,j,\delta) \in \set{S}_{i,j}: \delta \geq N\}.
\end{align}
Let $s_{i,j}^\textrm{b} = (i, j, N)$ denote the boundary state of error $(i, j)$.

Since the state space truncation does not affect the stationary probabilities of inner states, we have
\begin{align}
    \mu(s) = \mu^\prime(s), ~~\forall s \in \set{S}_{i, j}^\textrm{in}.
\end{align}
The boundary state $s_{i,j}^\textrm{b}$ absorbs the effect of all discarded states. Thus, we can write
\begin{align}
\mu^\prime(s_{i,j}^\textrm{b})=\sum\limits_{s\in\set{S}_{i,j}^\textrm{out}}\mu(s) =\sum_{k=0}^\infty \mu(s_{i,j}^\textrm{b}) (Q_{i, i}p_f)^k = \frac{\mu(s_{i,j}^\textrm{b})}{1-Q_{i, i}p_f}.
\end{align}
For error $(i, j)$, the performance gap can be computed as
\begin{align}
    \sigma_{i,j}(\pi, N) &= \set{L}_{i,j}(\pi) - \set{L}_{i,j}(\pi, N) \notag\\
    &= \sum_{s\in \set{S}_{i,j}^\textrm{out}}\mu(s) l(s, 1) - \mu^\prime(s_{i,j}^\textrm{b}) l(s_{i,j}^\textrm{b}, 1)\notag\\
    &= \sum_{s\in \set{S}_{i,j}^\textrm{out}}\mu(s) c(s) - \mu^\prime(s_{i,j}^\textrm{b}) c(s_{i,j}^\textrm{b})\notag\\
    &=\mu(s_{i,j}^\textrm{b})D_{i,j}\sum_{k=0}^\infty \epsilon_{i,j}(k),\label{eq:performance-gap}
\end{align}
where $\epsilon_{i,j}(k) \triangleq (Q_{i, i}p_f)^{k} \big(g_{i,j}(k+N) - g_{i,j}(N)\big)$.

By the ratio test, the summation $\sum_{k=0}^\infty \epsilon_{i,j}(k)$ is finite if
\begin{align}
    \lim_{k\rightarrow\infty}\frac{\epsilon_{i,j}(k+1)}{\epsilon_{i,j}(k)} = Q_{i, i}p_f\lim_{k\rightarrow\infty}\frac{g_{i,j}(k+N+1)}{g_{i,j}(k+N))} <1.
\end{align}
Therefore, if Assumption~\ref{assumption:age-functions} holds, the summation of the series $\{\epsilon_{i,j}(k)\}_{k\geq 0}$ is finite. Let $B_\textrm{max}$ be a finite constant such that $\sum_{k=0}^\infty \epsilon_{i,j}(k) \leq B_\textrm{max}$. It follows that
\begin{align}
    \sigma_{i,j}(\pi, N) \leq  \mu(s_{i,j}^\textrm{b})D_{i,j} B_\textrm{max}.\label{eq:performance-bound}
\end{align}
The stationary distribution of the boundary state $s_{i,j}^\textrm{b}$ is
\begin{align}
     \mu(s_{i,j}^\textrm{b}) = \mu(i, j, 1)(Q_{i, i})^{\tau_{i,j}-1}(Q_{i, i}p_f)^{N-\tau_{i, j}}.\label{eq:mu-b}
\end{align}
Substituting \eqref{eq:mu-b} into \eqref{eq:performance-bound} and taking a limit as $N\rightarrow \infty$ yields
\begin{align}
    0 < \lim_{N\rightarrow\infty}\sigma_{i,j}(\pi, N) \leq  D_{i,j} B_\textrm{max} \lim_{N\rightarrow\infty}\mu(s_{i,j}^\textrm{b})\rightarrow  0.
\end{align}
Then, the overall performance gap satisfies
\begin{align}
    \lim_{N\rightarrow\infty} \sigma(\pi, N)=\lim_{N\rightarrow\infty} \sum_{i\neq j} \sigma_{i,j}(\pi, N)\rightarrow 0.
\end{align}
Therefore, for any given switching policy $\pi$, the average costs satisfy the inequality $\lim_{N\rightarrow\infty}\set{L}(\pi, N) \rightarrow \set{L}(\pi)$. 

{\it{Step 2:}} Next, we rely on the squeeze theorem to prove the result~\cite{luo2024MobiHoc}. Let $\pi^*$ and $\pi^*(N)$ be the optimal policies for the original and truncated MDPs, respectively. Suppose $N$ is chosen such that the optimal thresholds are smaller than $N$. From step~1 we have $\set{L}(\pi) = \set{L}(\pi, N) + \sigma(\pi, N)$ for any switching policy $\pi$. It follows that
\begin{align}
    \set{L}(\pi^*(&N),N)
    + \sigma(\pi^*,N) \notag\\
    &\overset{(a)}{\leq} \set{L}(\pi^*,N)+ \sigma(\pi^*,N) \notag= 
    \set{L}(\pi^*) \notag\\
    &\overset{(b)}{\leq}  \set{L}(\pi^*(N)) = \set{L}(\pi^*(N),N)+ \sigma(\pi^*(N),N).
\end{align}
Herein, $(a)$ holds because $\pi^*(N)$ solves the truncated MDP, $(b)$ holds because $\pi^*$ solves the original MDP. Using the facts that $\set{L}(\pi^*(N),N) = \set{L}^*(N)$ and $\set{L}(\pi^*) = \set{L}^*$, we obtain
\begin{align}
    \set{L}^*(N)
    + \sigma(\pi^*,N) \leq
    \set{L}^*
    \leq \set{L}^*(N)+ \sigma(\pi^*(N),N).
\end{align}
Rearranging the terms yields
\begin{align}
    \set{L}^* - \sigma(\pi^*(N),N)
    \leq \set{L}^*(N) 
    \leq \set{L}^* - \sigma(\pi^*,N).
\end{align}

Taking a limit as $N \rightarrow \infty$ and using the fact that $\lim_{N\rightarrow\infty}\sigma(\pi,N) \rightarrow 0$, we obtain 
\begin{align}
     \lim_{N\rightarrow\infty} \set{L}^*(N) \rightarrow \set{L}^*.
\end{align}
This completes the proof.

\bibliographystyle{IEEEtran}
\bibliography{ref}

\begin{IEEEbiographynophoto}{Jiping~Luo} (Graduate Student Member, IEEE) received the B.S. degree in electronics and information engineering from Dalian Maritime University, Dalian, China, in 2020, and the M.S. degree in information and communication engineering from Harbin Institute of Technology, Shenzhen, China, in 2023. He is currently pursuing the Ph.D. degree with the Department of Computer and Information Science, Link\"oping University, Sweden. His research interests include semantic communications, networked control systems, and stochastic optimization.
\end{IEEEbiographynophoto}

\begin{IEEEbiographynophoto}{Nikolaos~Pappas} (Senior Member, IEEE) received the first B.Sc. degree in computer science, the second B.Sc. degree in mathematics, the M.Sc. degree in computer science, and the Ph.D. degree in computer science from the University of Crete, Greece, in 2005, 2012, 2007, and 2012, respectively. From 2005 to 2012, he was a Graduate Research Assistant with the Telecommunications and Networks Laboratory, Institute of Computer Science, Foundation for Research and Technology-Hellas, Heraklion, Greece; and a Visiting Scholar with the Institute of Systems Research, University of Maryland at College Park, College Park, MD, USA. From 2012 to 2014, he was a Post-Doctoral Researcher with the Department of Telecommunications, CentraleSupec, France. He is currently an Associate Professor with the Department of Computer and Information Science, Link\"oping University, Link\"oping, Sweden. His main research interests include the field of wireless communication networks, with an emphasis on semantics-aware communications, energy harvesting networks, network-level cooperation, age of information, and stochastic geometry. He has served as the Symposium Co-Chair for the IEEE International Conference on Communications in 2022. He is the general chair for the 23rd International Symposium on Modeling and Optimization in Mobile, Ad hoc, and Wireless Networks (WiOpt 2025). He is an Area Editor of the \textsc{IEEE Open Journal of the Communications Society} and an Expert Editor of invited papers of the \textsc{IEEE Communications Letters}. He is Associate Editor for four IEEE Transactions journals.
\end{IEEEbiographynophoto}

\end{document}